\documentclass[conference]{IEEEtran}
\usepackage{cite}
\usepackage{amsmath,amssymb,amsfonts}
\usepackage{graphicx}
\usepackage{textcomp}
\usepackage{xcolor}
\usepackage[utf8]{inputenc}
\usepackage{mathtools}

\usepackage{xspace}
\usepackage{adjustbox}
\usepackage{multicol}
\usepackage{ifthen}
\usepackage{xcolor}
\usepackage{chngcntr}
\usepackage{apptools}
\usepackage{tikz}
\usepackage{hyperref}
\usepackage{float}
\usepackage{subcaption}
\usepackage{algorithm, algpseudocode}

\def\BibTeX{{\rm B\kern-.05em{\sc i\kern-.025em b}\kern-.08em
    T\kern-.1667em\lower.7ex\hbox{E}\kern-.125emX}}
    
\def\checkmark{\tikz\fill[scale=0.4](0,.35) -- (.25,0) -- (1,.7) -- (.25,.15) -- cycle;} 
    
\newcommand{\fakeparagraph}[1]{\vspace{2mm}\noindent\textit{#1:}}
    
\newcommand{\thesystem}{Our-System\xspace}
\newcommand{\thesystemfull}{Our-System-full\xspace}
\newcommand{\thesystemcrippled}{Our-System-serial\xspace}

\floatstyle{ruled}
\newfloat{mod}{t}{mod}
\floatname{mod}{Abstraction}

\newtheorem{definition}{Definition}
\newtheorem{theorem}{Theorem}
\newtheorem{proof}{Proof}

\algblockdefx[<block>]{Upon}{EndUpon}[2][]{\textbf{upon} #2 \textbf{do}}{\textbf{done}}

\algblockdefx[<block>]{When}{EndWhen}[2][]{\textbf{when} #2 \textbf{do}}{\textbf{done}}

\newboolean{showcomments}
\setboolean{showcomments}{true}
\ifthenelse{\boolean{showcomments}}
{ \newcommand{\mynote}[3]{
    \fbox{\bfseries\sffamily\scriptsize#1}
    {\small$\blacktriangleright$\textsf{\emph{\color{#3}{#2}}}$\blacktriangleleft$}}}
{ \newcommand{\mynote}[3]{}}

\begin{document}%

%% \title{Kauri: Scaling Byzantine Fault Tolerant Consensus to large Trees}

\title{The quest for scaling BFT Consensus through Tree-Based Vote Aggregation}

\author{\IEEEauthorblockN{Ray Neiheiser, Miguel Matos, Lu\'{\i}s Rodrigues}
\IEEEauthorblockA{\textit{INESC-ID \& Instituto Superior T\'{e}cnico, Universidade de Lisboa}}
%% \IEEEauthorblockA{\textit{dept. name of organization (of Aff.)} \\
%% \textit{name of organization (of Aff.)}\\
%% City, Country \\
%% email address or ORCID
%% }
%% \and
%% \IEEEauthorblockN{2\textsuperscript{nd} Given Name Surname}
%% \IEEEauthorblockA{\textit{dept. name of organization (of Aff.)} \\
%% \textit{name of organization (of Aff.)}\\
%% City, Country \\
%% email address or ORCID}
%% \and
%% \IEEEauthorblockN{3\textsuperscript{rd} Given Name Surname}
%% \IEEEauthorblockA{\textit{dept. name of organization (of Aff.)} \\
%% \textit{name of organization (of Aff.)}\\
%% City, Country \\
%% email address or ORCID}
%% \and
%% \IEEEauthorblockN{4\textsuperscript{th} Given Name Surname}
%% \IEEEauthorblockA{\textit{dept. name of organization (of Aff.)} \\
%% \textit{name of organization (of Aff.)}\\
%% City, Country \\
%% email address or ORCID}
}

\maketitle

\thispagestyle{plain}

\begin{abstract}
With the growing commercial interest in blockchain, permissioned implementations have received increasing attention. Unfortunately, existing BFT consensus protocols that are the backbone of permissioned blockchains, either scale poorly or offer limited throughput. Most of these algorithms require at least one process to receive and validate the votes from all other processes and then broadcast the result, which is inherently non-scalable. Some algorithms avoid this bottleneck by using aggregation trees to collect and validate votes. However, to the best of our knowledge, such algorithms offer limited throughput and degrade quickly in the presence of faults. In this paper we propose \thesystem, the first BFT communication abstraction that organizes participants in a tree to perform scalable vote aggregation and that, in faulty runs, is able to terminate the protocol within an optimal number of reconfigurations ($f+1$). We define precisely which aggregation trees allow for optimal reconfiguration and show that, unlike previous protocols, when using these configurations, \thesystem scales to large number of processes and  outperforms HotStuff's throughput by up to 38x.
\end{abstract}


%\newpage

\section{Introduction}

As the range of blockchain use-cases expands to enterprise and government applications, permissioned blockchains, such as Hyperledger Fabric\cite{HyperledgerBlockchain}, have gained increasing attention. 
By having a restricted and well-known number of participants, permissioned blockchains can employ classical consensus protocols that have two main advantages over permissionless approaches: transaction finality and high throughput\cite{PoWVsBft}.

In this context, it is also relevant to design permissioned blockchains that are able to scale to hundreds of participants~\cite{sbft}. For instance,  Diem (formerly Libra) states that: ``Our goal was to choose a protocol that would initially support at least 100 validators and would be able to evolve over time to support 500–1,000 validators."\cite{libradocs}.

However, existing classical byzantine fault-tolerant consensus protocols scale poorly with the number of participants\cite{dls,byzGens}. This limitation stems from the large number of messages that need to be received, sent, and processed by a single process to reach consensus. For instance, the well-known PBFT protocol\cite{pbft} organizes participants in a clique resulting in a quadratic message complexity. More recent algorithms, such as HotStuff\cite{hotstuff}, organize participants in a star topology, where the leader sends/collects  information directly to/from all other participants, which reduces the message complexity but not the load on the coordinator, that still has to receive and validate votes from all other processes. These designs are, therefore, inherently non-scalable. 

An effective strategy to distribute the load and thus improve scalability is organizing the participants in a tree in combination with cryptographic schemes such as multi signatures\cite{byzcoin} or aggregated signatures\cite{bls}. 
This way, votes can be verified and aggregated as they are relayed up the tree to the leader, effectively spreading the computational load over the internal nodes. The leader can then verify and aggregate the result at a fraction of the computational cost, and disseminate it down the tree, evenly distributing bandwidth usage.
%An effective strategy to distribute the load, and therefore increase the scalability of the system, consists in organizing the participants in a logical tree, that can be used by the leader both to aggregate votes and to disseminate the consensus result. 
%By combining the tree topology with cryptographic schemes such as multi signatures\cite{byzcoin} or aggregated signatures\cite{bls}, each node in the tree may aggregate the votes of its children before propagating them to its parent.
%This allows votes to be collected and aggregated in parallel, as votes are propagated from the leaf to the leader at the root of the tree. 
%These cryptographic schemes allow the leader to validate the aggregated votes at a fraction of the cost of validating every vote individually.
Systems such as Byzcoin\cite{byzcoin} and Motor\cite{motor} leverage this design to offer better load balancing properties than the centralized approach proposed by HotStuff.

Unfortunately, the use of trees makes the system more sensitive to faults.
While in systems relying on a clique or star, such as PBFT or HotStuff, respectively, only the failure of the leader results in a reconfiguration or view change (under stable network conditions), in a tree a faulty internal node other than the leader also may prevent consensus from terminating.
In the former, upon failure of the leader, and if the network is stable, for $f$ faulty processes we need to perform at most $f+1$ reconfigurations to find a correct leader and hence ensure termination.
In contrast, in the latter, we need to find a tree with no faulty internal nodes, which may require a number of steps that is exponential with the system size~\cite{motor}.

%protocol inherently less robust, as the failure of a single internal node may prevent aggregation from terminating. Systems relying on a clique or star topology, such as PBFT or HotStuff, respectively, just have to perform at most $f+1$ reconfigurations to ensure termination (when the network is stable), for $f$ faulty nodes, which is optimal. However, in a tree topology, it is necessary to find a tree with no faulty internal nodes. This may require an exponential number of steps with respect to the system size~\cite{motor}.
Note that, while the reconfiguration takes place, the system is essentially halted which has a significant negative impact on the application. 
To mitigate the impact of a lengthy reconfiguration, systems such as Byzcoin or Motor~\cite{byzcoin,motor} fall back to a star topology after failing to terminate within a certain time, allowing the application to function in a degraded state.
While this might be tolerable for short periods, these approaches do not provide a strategy to reconstruct the tree and get back to the original performance and load balancing levels.
Given that in a large-scale system faults are the norm rather than the exception~\cite{failuresarecommon}, such deployments will inevitably fall back to a degraded state and stay there forever which is highly undesirable. In Byzcoin, for example, a single faulty node has a 50\% chance of leading the system into a degraded state.

In this paper, we propose \thesystem, the first communication abstraction for BFT consensus that organizes participants in a tree and, upon failures, guarantees termination after $f+1$ reconfigurations in the worst case, which is optimal. We precisely define which trees allow optimal reconfiguration, and discuss the tradeoffs between resilience, tree fan-out, and the number of reconfigurations required to ensure termination.

To leverage the capacity of non-leader processes and the parallelization opportunities made possible by a tree, \thesystem uses a novel pipelining scheme that allows nodes at different depths of the tree to aggregate votes for different consensus instances in parallel. As we show in the evaluation, this results in an up to 7x increase in throughput when compared to the non-pipelined version. We have implemented \thesystem and evaluated it under different realistic scenarios with up to 400 processes.
Results show that \thesystem outperforms the star topology in a range of different deployments, and outperforms HotStuff's throughput by up to 38x.

In short, the paper makes the following contributions:

\begin{itemize}
\item We present a set of abstractions that support the use of aggregation/dissemination trees in the context of consensus protocols;
\item We present a precise characterization of trees that support optimal reconfiguration steps;
\item We show how pipelining can be used to fully leverage the parallelisation opportunities offered by the trees;
%\item We present \thesystem, the first consensus protocol based on aggregation-trees that achieves termination with optimal number of reconfigurations;
\item We present \thesystem, the first tree- based communication abstraction for BFT consensus protocols that achieves termination in an optimal number of reconfigurations;
\item We present an extensive experimental evaluation of \thesystem in realistic scenarios with up to 400 nodes.

\end{itemize}

%The rest of this paper is organized as following. First, in Section~\ref{sec:relatedwork} we discuss the related work and how our approach compares to the state of the art. Section~\ref{sec:systemmodel} describes the system model and Section~\ref{sec:trees} starts outlining \thesystem. Next, Section~\ref{sec:reconfiguration} and Section~\ref{sec:pipelining} describe how the proposed system deals with the drawbacks of tree structures. Following that, Section~\ref{sec:implementation} describes the implementation of \thesystem that is evaluated in Section~\ref{sec:evaluation}. Section~\ref{sec:conclusion} concludes the paper.

\section{Related Work}
\label{sec:relatedwork}

The problem of Byzantine agreement was originally discussed in the context of synchronous systems\cite{byzGens} where it was solved using a recursive algorithm with factorial message complexity. Most Byzantine fault tolerant consensus protocols used today have been inspired by PBFT~\cite{pbft}. PBFT guarantees safety during asynchrony periods but requires synchronous phases to guarantee progress. In PBFT processes use an \textit{all-to-all} communication scheme which offers, for a small numbers of processes, excellent throughput and latency. However, due to the quadratic message cost it does not scale to large system sizes. Although there have been many proposals to extend and improve several aspects of PBFT (for instance, \cite{ebawa, Archer, bftsmart, spinning}) most preserve the same communication pattern.

Meanwhile, the recent interest in permissioned blockchains, where consensus plays a crucial role, and the need to scale such systems to thousands of processes\cite{sbft,libradocs}, fueled the research on alternative algorithms, with better scalability and load balancing properties.

HotStuff organizes processes in a star-topology to reduce message complexity\cite{hotstuff}. It requires four communication rounds where, in each round, the process at the center of the star plays the role of the leader that collects and disseminates votes to and from all other processes. While the protocol has a linear message complexity, it creates a bottleneck at the leader, that has to communicate with all processes and verify a quorum of $2f+1$ signatures each round. At the time of writing, the publicly available implementation of HotStuff uses \textit{secp256k1}~\cite{libsec}, a highly efficient elliptical curve algorithm, also used in Bitcoin, but against which several attack vectors have been found\cite{savecurves}. In this implementation, the leader is required to relay the full set of signatures to all processes. Alternatively, it is possible to use threshold signatures such as \textit{bls}~\cite{bls}, to reduce the message size at the expense of a significant increase in the computational load at the leader. For instance, \textit{bls}  signatures are up to 30 times slower to verify than \textit{secp256k1} signatures (see Section~\ref{sec:evaluation}). 

Another interesting feature of HotStuff is the use of a \textit{pipelining} optimization that allows to overlap different rounds of different consensus instances. Since consensus requires multiple rounds of communication, HotStuff can optimistically start the consensus for the next block before the consensus for the current one is terminated. This allows up to four consensus instances to run in parallel resulting in higher throughput.

Systems such as Steward~\cite{steward} and Fireplug~\cite{fireplug} organize processes in hierarchical groups resulting in low message complexity and better load balancing. However, these systems still rely on an all-to-all communication pattern in the multicast group positioned at the root of the hierarchy and, therefore, to ensure scalability, make restrictive assumptions regarding the distribution of faults.

Another alternative is to organize processes in a tree topology. Byzcoin~\cite{byzcoin} organizes processes in a binary tree where signatures (votes) are aggregated as they are relayed up the tree from the leaf nodes to the root, and the results are disseminated from the root to the leaves. While this design results in linear message complexity and has good load balancing properties, the failure of a single internal node may prevent termination. Byzcoin overcomes this difficulty by falling back to a clique topology (which has quadratic message complexity) when consensus fails.

Motor~\cite{motor} and Omniledger~\cite{omniledger} build upon the principles of Byzcoin, but rather than falling back immediately to a clique topology, rely instead on the use a scheduler that rotates the nodes in the subtrees until the leader is able to gather $N-f$ signatures. If the root process is unable to collect sufficient signatures, it contacts directly a random subset of  leaf processes, which in turn will attempt to collect votes from their siblings, until a quorum is obtained. Thus, in the worst case, assuming a fan-out of $m$, after $\frac{N}{m}$ steps the tree will have collapsed to a star topology. Further, in the presence of a faulty leader a total of $(f + 1) * \frac{N}{m}$ steps might be necessary until liveness is reestablished. For this to work, the tree must be limited to a depth of two, otherwise finding a tree that ensures liveness requires an exponential number of steps.

Strikingly, neither Motor nor Omniledger propose mechanisms to reconstruct the tree after falling back to the degraded topology. Thus, when faults occur, these systems lose their desirable communication properties permanently. In a large scale system, where failures are common~\cite{failuresarecommon}, this implies that those systems are more likely to run in a degraded state than in the optimal one. 

Another disadvantage of trees is the increased number of communication steps per round of communication which negatively impacts latency. If consecutive consensus instances are executed in serial order, this increase in latency also hurts throughput. This can be mitigated through pipelining techniques such as the ones proposed by HotStuff.

Table~\ref{tab_comparison} summarizes the properties of the systems described above. Byzcoin, Motor, and Omniledger are able to distribute the processing load, because signatures are aggregated on their path to the leader, such that the leader does not need to verify all signatures. For the same reason, these systems distribute bandwidth usage among the root and the internal nodes, such that there is no single process that has to communicate with every other process. Unfortunately, neither of them is able to converge to a configuration that ensures termination in an optimal number of reconfiguration steps. Also, none of these systems implements pipelining, which is key to mitigate the impact of the extra-latency on the system throughput.

\begin{table}[t]
\centering
{\scriptsize
\caption{Comparison of existing Algorithms}
\label{tab_comparison}
\begin{tabular}{|l|c|c|c|c|}
\hline
 & \textbf{Balances} & \textbf{Balances}    &  \textbf{$t+1$} & \textbf{Supports} \\ 
 & \textbf{processing} & \textbf{bandwidth}   &  \textbf{Reconfig.} & \textbf{Pipelining}\\  
 & \textbf{load} & \textbf{consumption}   &  \textbf{steps} &  \\ 
\hline
 HotStuff  & $\bullet$  & $\bullet$   & \checkmark & \checkmark \\ 
 Steward/ Fireplug & \checkmark  & \checkmark  & $\bullet$ & $\bullet$ \\ 
 Byzcoin\ & \checkmark   & \checkmark & $\bullet$  & $\bullet$ \\ 
 Motor & \checkmark  & \checkmark  & $\bullet$  & $\bullet$ \\
 Omniledger & \checkmark  & \checkmark  & $\bullet$  & $\bullet$ \\\hline
 \thesystem (this paper) & \checkmark  & \checkmark & \checkmark & \checkmark \\
\hline
\end{tabular}
}
\end{table}

In this paper, we propose \thesystem, a system aimed at overcoming the scalability limitations found in previous works. \thesystem owns the advantages of Byzcoin, Motor, and Omniledger, but without their disadvantages. Like these systems, \thesystem organizes the communication pattern on top of a tree to distribute both the processing load and the bandwidth utilization. In contrast to these systems, \thesystem implements a reconfiguration strategy that produces a robust tree in optimal (i.e. $f+1$) steps. We name this property Optimal Conformity and formally define it in Section~\ref{sec:trees}. Also, \thesystem leverages pipelining to hide the latency effects induced by the tree. Instead of proposing our own consensus algorithm, \thesystem can be plugged into any leader-based consensus algorithm to replace the communication scheme resulting in increased throughput and improved load balancing.

\section{System Model}
\label{sec:systemmodel}

The system is composed of a set of $N$ server processes $\{p_1, p_2, .. p_N\}$ and a set of client processes $\{c_1, c_2, .. c_m\}$. Client and server processes are connected through perfect point-to-point channels (constructed by adding mechanisms for message re-transmission as well as detecting and suppressing duplicates~\cite{cachinbook}). We also assume the existence of a Public Key Infrastructure used by processes to distribute the keys required for authentication and message signing.  Moreover, processes may not change their keys during the execution of the protocol and require a sufficiently lengthy approval process to re-enter the system to avoid rogue key attacks~\cite{roguekey}.

We assume the Byzantine fault model, where at most $f \leq \frac{N-1}{3}$ faulty processes may produce arbitrary values, delay or omit messages, and collude with each other, but do not possess sufficient resources to compromise cryptographic primitives. 

To circumvent the impossibility of consensus\cite{flp}, we assume the partial synchrony model~\cite{dls}. In this model, there may be an unstable period, where messages exchanged between correct processes may be arbitrarily delayed. Thus, while the network remains unstable, it is impossible to distinguish a faulty process from a slow one. However, there is some known bound $\Delta$ on the worst-case network latency and an unknown Global Stabilization Time (GST), such that after GST, all messages between correct processes arrive within $\Delta$.

\section{From a Star to a Tree Topology}
\label{sec:trees}

Instead of designing a completely new consensus algorithm from scratch, we developed \thesystem as an adaptation of HotStuff. The key idea is to replace the dissemination and aggregation patterns used by HotStuff, which  are based on a star topology, by new patterns based on  trees.
While our presentation hinges on HotStuff characteristics, our principles could be applied to any other leader-based consensus algorithm.

\subsection{HotStuff Communication Pattern}

Due to space constraints, it is not possible to provide a detailed description of HotStuff. For self-containment, we provide a brief high-level description of its operation. We give emphasis on the communication pattern used in HotStuff and discuss how this pattern may be abstracted, such that it can be replaced by a more scalable implementation. HotStuff reaches consensus in four communication rounds. Each round consists of two phases: i) a dissemination phase where the leader broadcasts some information to all processes; and ii) an aggregation phase where the leader collects and aggregates information that it received from a quorum of replicas. All rounds follow the same exact pattern, but the information sent and received by the leader in each round differs:

\fakeparagraph{First round} In the dissemination phase, the leader broadcasts a block proposal to all processes. In the aggregation phase, the leader collects a \textit{prepare} quorum of $2f+1$ signatures of the block. The signatures convey that  the replicas have validated and accepted the block proposed by the leader. 

\fakeparagraph{Second round} In the dissemination phase, the leader broadcasts the \textit{prepare} quorum, either as a collection of individual signatures or as a single combined signature (when threshold or aggregated signatures are used). In the aggregation phase, the leader collects the \textit{pre-commit} quorum, a quorum of signatures which imply that the processes have validated the \textit{prepare} quorum. If the leader is able to collect a quorum, the value proposed by the leader is \textit{locked} and will not be changed, even if the leader is suspected and a reconfiguration occurs.

\fakeparagraph{Third round} In the dissemination phase, the leader broadcasts the \textit{pre-commit} quorum. In the aggregation phase, the leader collects a \textit{commit} quorum, a quorum of signatures implying that processes have validated the \textit{pre-commit} quorum. If the leader is able to collect the quorum, the value is decided.

\fakeparagraph{Forth round} In the last round, the leader disseminates the \textit{commit} quorum in a decision message to all processes, which then verify it and decide accordingly.

\begin{figure}
	\centering
	\subfloat[Topology\label{fig:examplestar}]{\includegraphics[width=0.24\linewidth]{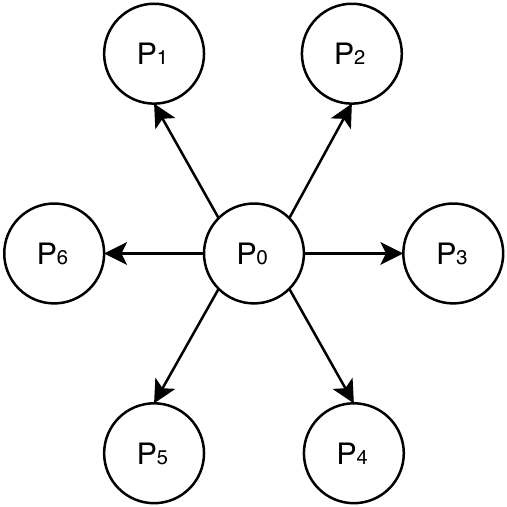}}\,\,
	\subfloat[Communication pattern\label{fig:threephasehotstuf}]{\includegraphics[width=0.5\linewidth]{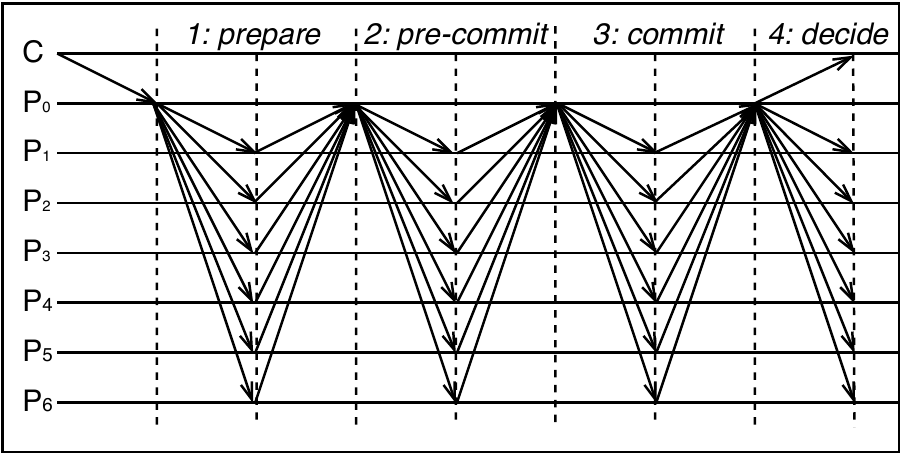}}
    \caption{HotStuff communication pattern with seven processes.}
    \label{fig:star}
\end{figure}

\vspace{2mm}Figure~\ref{fig:star} illustrates the communication pattern of HotStuff in a system with 7 processes. We abstract the communication pattern used in HotStuff with the following two primitives:

\begin{itemize}
\item \textsl{disseminate-to-followers(data)}. This primitive is invoked by the leader; it is used in the first phase of each round to send \textsl{data} to all other processes. 
\item \textsl{aggregate = aggregate-at-leader(input$_i$, input$_j$, $\ldots$)}. This primitive is invoked by every non-faulty process $i$ that received \textsl{data} from the leader in the first phase. The input by process $i$ consist of a cryptographic signature of \textsl{data}. The primitive returns an aggregate of the input signatures to the leader, containing at least $2f+1$ input values. 
\end{itemize}

While HotStuff uses one implementation of the above primitives, other implementations are possible as long as they offer the same guarantees. To characterize the guarantees provided by the communication pattern of HotStuff, we first define the notion of a \textit{robust configuration}. We say that the primitives are executed in a robust configuration if the leader is correct and in a non-robust configuration if the leader is faulty. Using this notion, we can capture the properties of the \textsl{disseminate-to-followers} and \textsl{aggregate-at-leader} primitives, respectively:

\begin{definition} 
\textbf{Reliable Dissemination:} After the Global Stabilization Time, and in a robust configuration, all correct processes deliver the \textsl{data} sent by the leader.
\end{definition}

\begin{definition} 
\textbf{Fulfilment:} After the  Global  Stabilization Time, and in a robust configuration, the aggregate collected by the leader includes at least $2f+1$ signatures.
\end{definition}

If the configuration is not robust, the primitives offer no guarantees. In particular, some nodes may receive the same \textsl{data} from the leader and others may receive different or no \textsl{data} at all. The aggregate arriving at the leader may include less than $2f+1$ valid signatures, even if the leader sends consistent and valid \textsl{data} to all processes.
Note that, before the Global Stabilization Time is reached, it may be impossible to distinguish a robust from a non-robust configuration, \emph{i.e.} whether the leader is correct or faulty.

\subsection{Using Trees to Implement HotStuff}

We now discuss how to implement the \textsl{disseminate-to-followers} and \textsl{aggregate-at-leader} primitives using tree topologies, while preserving the same properties. As noted before, nodes are organized in a tree with the leader at the root. The primitive \textsl{disseminate-to-followers} is implemented by having the root send \textsl{data} to its children, that in turn, forward it to their own children, and so forth. The primitive \textsl{aggregate-at-leader} is implemented by having the leaf nodes send their signatures to their parent. The parent then aggregates those signatures with its own and sends the aggregate to their own parent. This process is repeated until the final aggregate is computed at the root of the tree. This process is illustrated in Figure~\ref{fig:tree}.

\begin{figure}
	\centering
	\subfloat[Topology\label{fig:exampletree}]{\includegraphics[width=0.24\linewidth]{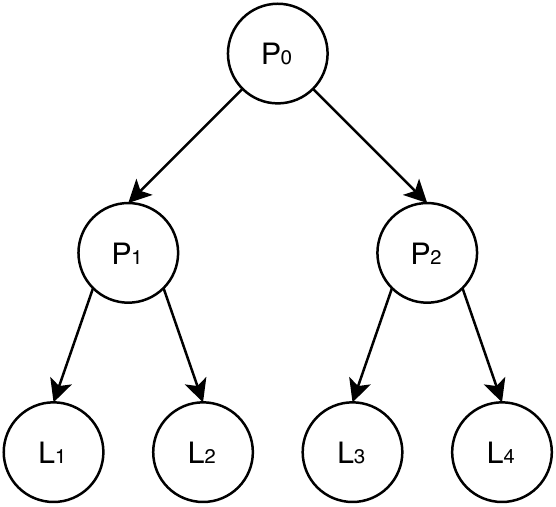}}\,
	\subfloat[Communication pattern\label{fig:threephase}]{\includegraphics[width=0.74\linewidth]{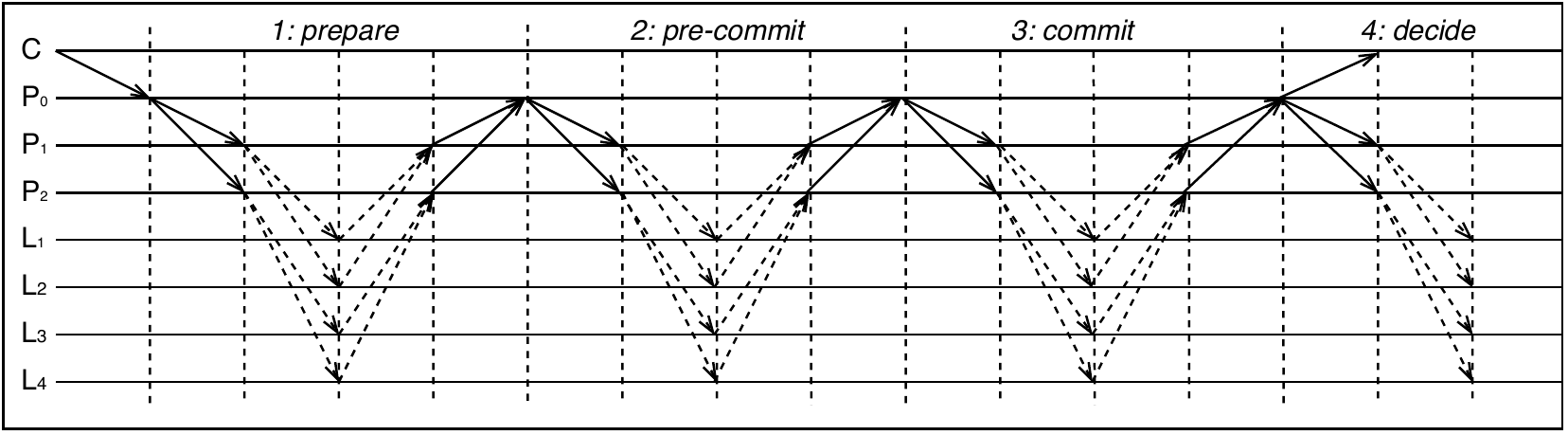}}
    \caption{Tree communication pattern for seven processes.}
    \label{fig:tree}
\end{figure}

When using a tree to implement \textsl{disseminate-to-followers} and \textsl{aggregate-at-leader}, the notion of robust configuration needs to be adapted. In fact, in a tree, it is not enough that the leader is non-faulty to make the configuration robust. Instead, the communication graph (in our case, the communication tree), must be robust. We define a \textit{robust graph} as follows.

\begin{definition}
\textbf{Robust Graph:} An edge is said to be \textit{safe} if the corresponding vertices are both correct processes. A graph is robust iff the leader process is correct and, for every pair of correct process $p_i$ and $p_j$, the path in the tree connecting these processes is composed exclusively of safe edges\label{def:robust}.
\end{definition}

\subsection{Achieving Reliable Dissemination and Fulfillment}

In the following we describe how to implement \textsl{disseminate-to-followers} and \textsl{aggregate-at-leader} primitives on a tree such that Reliable Dissemination and Fulfillment are satisfied. We start by describing the communication primitives used to propagate information on the tree and the cryptographic primitives used to perform aggregation.

\subsubsection{Communicating on the Graph}
Processes use the tree to communicate. 
%Each edge in the graph corresponds to a perfect point-to-point channel that is used by neighbor processes to communicate.  
Each directed edge maps to a perfect single-use point-to-point channel, to send and deliver a single value. Perfect point-to-point channels expose a \textsc{send} primitive used to send a message, and a \textsc{deliver} primitive used to receive a message, and exhibit the following properties.

\begin{itemize}
    \item \emph{Validity:} If a process $p_j$ delivers a value $v$ on a channel over an edge $e_{ij}$, $v$ was sent by $p_i$.
    \item \emph{Termination:} If both $p_i$ and $p_j$ are correct, if $p_i$ invokes \textsc{send} then eventually $p_j$ delivers $v$.
\end{itemize}

%When using the overlay to aggregate information, nodes wait for their children to propagate data upstream. Conversely, when the overlay is used to dissiminate information, nodes wait for their parents to propagate downstream. For liveness, the aggregation/dissemination process must proceed even in the occurrence of faults. 
Note that, when using perfect channels, a message is only guaranteed to be eventually delivered if both the sender and the recipient are correct. If the sender is faulty, no message may ever be delivered. To avoid blocking, a process should be able to make progress if a message takes too long to be received. This behaviour is captured by an abstraction we call \textit{impatient channels}. Impatient channels offer a blocking \textsc{receive} primitive that always returns a value: either the value sent by the sender, or a special value $\bot$ if the sender is faulty or the system is unstable. After the Global Stabilization Time (GST)\cite{dls}, if the sender and the receiver are correct, the receiver always receives the value sent. Impatient channels have the following properties:

\begin{itemize}
    \item \emph{Validity:} 
    %If a process $p_j$ delivers a value $v$,  $v$ was sent by process $p_i$ or $v=\bot$.
    If a process $p_j$ delivers a value $v$ on a channel over an edge $e_{ij}$, $v$ was sent by $p_i$ or $v=\bot$.

    \item \emph{Termination:}  If a correct process $p_j$ invokes \textsc{receive}, it eventually returns some value.
    \item \emph{Conditional Accuracy:} Let $p_i$ and $p_j$ be correct sender and receiver processes, respectively. After GST, $p_j$ always return the value $v$ sent by $p_i$.
\end{itemize}

Algorithm~\ref{algo:impchannels-receive} shows how impatient channels can be implemented on top of perfect channels using the known bound $\Delta$ on the worst-case network latency.

\setlength{\textfloatsep}{0pt}
\begin{algorithm}[t]
\scriptsize
\centering
\caption{Impatient Channels: implementation of \textsc{receive}}
\label{algo:impchannels-receive}
\begin{algorithmic}[1]
\State{let \textit{ic} be an impatient channel built on top of perfect channel \textit{pc}}
\Function{\textit{ic}\textsc{.receive}}{ }
\State{\textsc{timer.start}($\Delta$)}
\State{\textbf{when} \textit{pc}.\textsc{deliver} ($v_i$) \textbf{do return} $v_i$}
\State{\textbf{when} \textsc{timer.timeout}() \textbf{do return} $\bot$}
\EndFunction
\end{algorithmic}
\end{algorithm}
\setlength{\textfloatsep}{5pt}

\subsubsection{Cryptographic Collections}

In each round of consensus, it is necessary to collect a Byzantine quorum of votes. The collection and validation of these votes can be an impairment for scalability. \thesystem mitigates these costs by using the tree to aggregate votes as they are forwarded to the leader. We model the process of vote aggregation with an abstraction called a \textit{cryptographic collection}, that corresponds to a secure multi-set of tuples $(p_i, v_i)$. A process $p_i$ can create a new collection $c$ with a value $v_i$ by calling $c$=\textsc{new}($(p_i, v_i)$). Processes can also merge two collections using a \textsc{combine} primitive,  $c_{12}$ = \textsc{combine}$(c_1, c_2)$, denoted by $c_{12} = c_1 \oplus c_2$. A process can also check if a collection $c$ includes at least a given threshold of $t$ distinct tuples with the same value $v$, by calling \textsc{has}$(c, v, t)$. Finally, it is possible to check the total amount of input tuples that were contributed to $c$ by checking its cardinality $\vert c \vert$. Cryptographic collections have the following properties:

\begin{itemize}
    \item \emph{Commutativity:}   $c_1 \oplus c_2 = c_2 \oplus c_1$.
    \item \emph{Associativity:}   $c_1 \oplus (c_2 \oplus c_3)  = (c_1 \oplus c_2) \oplus c_3$.
    \item \emph{Idempotency:} $c_1 \oplus c_1 = c_1$.
    \item \emph{Integrity:} Let $c= c_1 \oplus \ldots c_i \ldots c_n$. If \textsc{has} ($c, v, t$) then at least $t$ distinct processes $p_i$ have executed $c_i$=\textsc{new}($(p_i, v)$).
\end{itemize}

Note that different cryptographic techniques can be used to implement these collections. In \thesystem, we leverage a non-interactive \textit{bls} cryptographic aggregation scheme that allows each internal node to aggregate the votes from its children into a single aggregated vote~\cite{bls}. The complexity of verifying an aggregated vote is $\mathcal{O}(1)$, thus the burden imposed on each internal node (including the root) is $\mathcal{O}(m)$, where $m$ is the fanout of the tree.  Note that classical asymmetric signatures require $\mathcal{O}(N)$ verifications at every process and threshold signature schemes require $\mathcal{O}(N)$ at the root and $\mathcal{O}(1)$ at all the other processes~\cite{bls}.

\subsubsection{Implementing disseminate-to-followers}

Algorithm~\ref{algo:disseminate} presents the implementation of \textsl{disseminate-to-followers} on a tree. 
Despite being very simple, note that the algorithm always terminates, even if some intermediate nodes are faulty. This is guaranteed since impatient channels always return a value after the known bound $\Delta$ on the worst-case network latency, either the \textsl{data} sent by the parent or the special value $\bot$.

\setlength{\textfloatsep}{0pt}
\begin{algorithm}[t]
\scriptsize
\centering
\caption{\textsl{disseminate-to-followers} on a tree $G$ (process $p_i$)}
\label{algo:disseminate}
\begin{algorithmic}[1]
\Procedure{\textsl{disseminate-to-followers}}{$G$, \textsl{data}}
\State{$\textit{children} \gets $ $G$\textsc{.children}($p_i$)} \Comment{Get edges to children of $p_i$}
\State{$\textit{parent} \gets $ $G$\textsc{.parent}($p_i$)} \Comment{Get parent of $p_i$ (returns $\bot$ for root)}
\If{ $\textit{parent} \neq \bot$}
%\For{\textbf{all} $\textit{e} \in \textit{parent}_i$} %\Comment{Empty for root}
    \State {\textsc{channel.receive}(\textit{parent}, \textsl{data})}
    \Comment{Receive from parent}
%\EndFor

\EndIf
\For{\textbf{all} $\textit{e} \in \textit{children}$} \Comment{Send to children}
    \State {\textsc{channel.send}($e$, $\textit{data}$)}
\EndFor
\State{\textbf{return} $\textsl{data}$}
\EndProcedure
\end{algorithmic}
\end{algorithm}
\setlength{\textfloatsep}{5pt}

\begin{theorem}
Algorithm~\ref{algo:disseminate} guarantees Reliable Dissemination.
\end{theorem}

\begin{proof}
We prove this by contradiction. Assume Reliable Dissemination is not guaranteed.
This implies that at least one correct process did not receive the data sent by the leader. This is only possible if: i) at least one correct process is not connected to the leader either directly or through intermediary correct processes ii) one of the intermediary processes or the root process did not invoke \textsc{channel.send} for at least one correct child process, or iii) the data got lost in the channel.
By definition, Reliable Dissemination requires a robust configuration which following the definition of a Robust Graph ensures that the leader is correct and there is a path of correct processes between the leader and any other correct process. Thus, the first case is not possible.
Moreover, correct processes follow the algorithm and, because correct processes can only have correct parents in a robust configuration, the second case is also impossible. Finally, the third case is also impossible due the perfect channels.
%However, that directly leads to a contradiction. Since a robust graph is assumed, from 
%that the leader process must be correct and there must be a path of correct processes between the correct leader and any other correct process. Additionally, we assume perfect communication channels making the loss of data impossible.
Therefore, Algorithm~\ref{algo:disseminate} guarantees Reliable Dissemination.
\end{proof}

\subsubsection{Implementing aggregate-at-leader}

Algorithm~\ref{algo:aggregate} presents the implementation of \textsl{aggregate-at-leader} on a tree. The algorithm relies on the cryptographic primitives to aggregate the signatures as they are propagated towards the root. Like \textsl{disseminate-to-followers}, \textsl{aggregate-at-leader} always terminates, even if some intermediate nodes are faulty. This is guaranteed because impatient channels always return a value after the known bound $\Delta$ on the worst-case network latency, either the \textsl{data} sent by the child processes or the special value $\bot$. However, before GST or due to a non-robust configurations, the collection returned at the leader may be empty or include just a subset of the required signatures.

\setlength{\textfloatsep}{0pt}
\begin{algorithm}[h!]
\scriptsize
\centering
\caption{\textsl{aggregate-at-leader} on a tree $G$ (process $p_i$)}
\label{algo:aggregate}
\begin{algorithmic}[1]
\Procedure{\textsl{aggregate-at-leader}}{$G$, \textsl{input}}
\State{$\textit{children} \gets $ $G$\textsc{.children($p_i$)}} \Comment{Get edges to children of $p_i$}
\State{$\textit{parent} \gets $ $G$\textsc{.parent}($p_i$)}\Comment{Get parent of $p_i$ (returns $\bot$ for root)}
\State{$\textit{collection} \gets \textsc{new}((p_i, \textsl{input}))$}
\For{\textbf{all} $\textit{e} \in \textit{children}$} \Comment{Empty for leaf nodes}
    \State {\textsc{channel.receive}($e$, \textit{partial})} \label{algo:aggregate-receive}
    \State {$\textit{collection} \gets \textit{collection} \oplus \textit{partial}$}
\EndFor
\If{ $\textit{parent} \neq \bot$}
%\For{\textbf{all} $\textit{e} \in \textit{parent}_i$} \Comment{Send to  parent, if not the root}
    \State {\textsc{channel.send}(\textit{parent}, $\textit{collection}$)} \label{algo:aggregate-send}
\EndIf
%\EndFor
\State{\textbf{return} $\textit{collection}$}
\EndProcedure
\end{algorithmic}
\end{algorithm}
\setlength{\textfloatsep}{5pt}

\begin{theorem}
Algorithm~\ref{algo:aggregate} guarantees Fulfilment.
\end{theorem}

\begin{proof}
We prove this by contradiction. Assume that the leader process was unable to collect $2f+1$ signatures. 
Following Algorithm~\ref{algo:aggregate}, this means that either: i) an internal node did not receive the signatures from all correct children (line~\ref{algo:aggregate-receive}), ii) or an internal node did not aggregate and relay the signatures (line~\ref{algo:aggregate-send}). Since we assume impatient channels, that are implement on top of perfect point-to-point channels,  the first case is not possible. The second case may happen, if either the internal node is not a correct process and omits signatures in the aggregate, if it does not relay any signatures, or if it is waiting indefinitely for all answers of its child processes.
Either option leads to a contradiction. Since we assume a robust graph, all internal nodes between the root and a correct process must be correct and hence follow the protocol. Additionally, due to the impatient channels, eventually each channel will return a value to the internal node making sure that eventually it will receive an answer from all child processes (correct or not) allowing it to relay all collected signatures from all correct child processes.
Therefore Algorithm~\ref{algo:aggregate} guarantees Fulfilment.
\end{proof}

\subsubsection{Challenges of Using a Tree}

With the implementation of \textsl{disseminate-to-followers} and \textsl{aggregate-at-leader} introduced above, it is possible to have a tree topology that offers the same properties as the star topology of HotStuff.
%when running in a robust configuration, the tree topology offers the same guarantees as the star-pattern used in HotStuff. 
However, two remaining challenges need to be addressed to make the tree topology a valid alternative:

\fakeparagraph{Reconfiguration strategy} In HotStuff, the configuration is robust if the leader is non-faulty. Therefore, there are only $f$ non-robust configurations and $N-f$ robust configurations. 
It is thus trivial to devise a 
a reconfiguration strategy that yields a robust configuration in optimal  steps (\emph{i.e. $f+1$}). When using a tree, a configuration is only robust if the root and all internal nodes are correct. 
The total number of configurations and, the subset of non-robust configurations, is extremely large. A reconfiguration strategy that is able to guide the system to a robust configuration in a small number of steps is therefore non-trivial.

\fakeparagraph{Mitigate the impact of the increased latency} While trees allow to distribute the load among all processes, it comes at the cost of increased latency of the \textsl{disseminate-to-followers} and \textsl{aggregate-at-leader} primitives which, in turn, may negatively affect the throughput of the system. Mitigating these effects requires additional mechanisms.

\vspace{2mm}In the next sections, we discuss how \thesystem addresses the two challenges above.

\section{Reconfiguration}
\label{sec:reconfiguration}

We now discuss \thesystem's reconfiguration strategy. A reconfiguration is necessary when, due to faults or due to an asynchronous period, the current configuration is deemed to be not robust.
Every reconfiguration results in a new tree that processes use to communicate. Naturally, not all the trees are robust and several reconfigurations might be necessary before a robust tree is found. 

\subsection{Modeling Reconfiguration as an Evolving Graph}

We model the sequence of trees resulting from the reconfigurations as an evolving graph, i.e., a sequence of static graphs (that are trees). To ensure that eventually a robust graph is used, the evolving graph must hold the following property:

\begin{definition} 
\textbf{Recurringly Robust Evolving Graph:} An evolving graph $\mathcal{G}$ is said to be \textit{recurringly robust} iff a robust static graphs appear infinitely often in its sequence. \label{def:recurringlyrobust}
\end{definition}

A recurringly robust evolving graph is sufficient to ensure that eventually a robust graph will be used by processes to communicate. However, this is undesirable in practice because the number of reconfigurations until a robust graph is found is unbounded. Because the system's process is essentially halted during reconfiguration, we would like to find a robust graph after a small number $t$ of reconfigurations. We call this property of an evolving graph \textit{$t$-bounded conformity}. 

\begin{definition}
\textbf{$t$-Bounded Conformity}: a recurringly robust evolving graph $\mathcal{G}$ exhibits \textit{$t$-bounded conformity} if a robust static graph appears in $\mathcal{G}$ at least once every $t$ consecutive static graphs.\label{def:boundedconformity}
\end{definition}

Interestingly, for general topologies it is hard to build a reconfiguration strategy that construct an evolving graph that satisfies $t$-bounded conformity for small values of $t$. Consider, for instance, an evolving graph where all static graphs are binary trees. The number of possible binary trees is given by the Catalan number $C_N = \frac{(2N)!}{((N + 1)!N!}$. From all these trees, only a small fraction is robust, namely those where faulty processes are not internal nodes\footnote{A tree with a faulty internal node might be robust if all its children are also faulty but, for simplicity, we omit these cases in the analysis.}.  
Thus, a naive evolving graph that simply permutes among all possible configurations may require a factorial number of steps to find a robust graph.

Ideally, we want an upper bound on the number of reconfiguration steps that grows linearly either with the system size $N$ or the number of faults $f$. Note that, in the worst case, any leader-based protocol may require at least $f+1$ reconfigurations, because a faulty leader may always prevent consensus from being reached. This leads us to the notion of \emph{Optimal Conformity}.

\begin{definition}
\textbf{Optimal Conformity}: a recurringly robust evolving graph $\mathcal{G}$ fulfills Optimal conformity iff $\mathcal{G}$ observes $(f+1)$-bounded conformity.
\label{def:optimalconformity}
\end{definition}

We now discuss how to achieve Optimal Conformity in tree topologies. To make the problem tractable, we restrict ourselves to a balanced tree topology with a fixed fanout. In this case, for a tree of $N$ nodes, of which $I$ are internal nodes (including the root), there are ${N \choose I}$ different relevant  assignments of processes to internal nodes in the tree. Of these assignments, only a small fraction yields a robust tree. This fraction is given by: 

{\small
\begin{equation}
\frac{(N-f)!(N-I)!}{(N-f-I)!N!}
\end{equation}\label{eq:pk1}
}
Therefore, if the reconfiguration strategy selects trees at random, the probability of obtaining a robust tree is very small. For example, when ${N\to\infty}$, for $I=4$ the probability to obtain a robust graph, is  $\approx 20\%$ and for $I=10$ it is  only $\approx 1.7\%$. This is the main reason why the state of the art approaches that rely on trees degrade to a star or clique topology upon failures. Next, we introduce a reconfiguration strategy that addresses this problem.

\subsection{Building an Evolving Tree with Optimal Conformity} 

Our construction is based on the following key insight: if we split all processes into $f+1$ disjoint bins then at least one of these bins is guaranteed to have only correct processes ((as there are at most $f$ faulty processes)). By construction, a tree whose internal nodes are exclusively drawn from this bin is guaranteed to be robust. For this strategy to work, we must have $\frac{N}{I} \geq f + 1$, i.e. the number of processes must be sufficient to fill $f+1$ bins of size $I$. Algorithm~\ref{algo:robustGraph} captures the steps required to construct an evolving tree that satisfies Optimal Conformity. It starts by dividing the set of processes in 
$f+1$ disjoint bins, each with at least $I$ nodes. Each static graph $G^k$ is built by picking a bin $B^i$, following a round robin strategy, and by assigning nodes from bin $B^i$ to all internal nodes of the tree (including the root). 

\begin{algorithm}[!t]
\scriptsize
\centering
\caption{Construction of an Evolving Tree with Optimal Conformity ($\frac{N}{I} \geq f + 1$)}
\label{algo:robustGraph}
\begin{algorithmic}[1]
\State{Initially, split the set of processes $\mathcal{N}$ into  disjoint bins}
\State{$\mathcal{N} \gets B^0 \cup B^1 \cup \ldots \cup B^{\frac{N}{I}}$; s.t. $|B^i| \geq f+1 \land B^i \cap B^j = \emptyset$. }
\Function{build}{$k$} 
\State{$i \gets k\;\textbf{mod}\;I$}\label{alg:e-tree:iteration}
\State{$\mathcal{G}^i \gets$ all possible trees whose internal nodes are drawn exclusively from $B^{i}$. }
\State{$G^k \gets$~pick any tree from $\mathcal{G}^i$}
\State{\textbf{return}($G^k$)}\label{alg:e-tree:return}
\EndFunction
\end{algorithmic}
\end{algorithm}

\begin{theorem}
\label{th:generictree}
Algorithm~\ref{algo:robustGraph} constructs an evolving graph that satisfies Optimal Conformity if $\frac{N}{I} \geq f + 1$.
\end{theorem}

\begin{proof} Algorithm~\ref{algo:robustGraph} divides all processes into disjoint bins of size $I$ for a total of $\frac{N}{I}$ bins. As long as $\frac{N}{I} \geq f + 1$ there are at least $f+1$ bins and, therefore, there are at most $f$ bins containing faulty processes. 
Upon reconfiguration, the algorithm picks a tree whose internal nodes are drawn from the next bin, hence guaranteeing that a robust static graph is found after at most $f+1$ steps.
    %If the system is reconfigured by picking the next static graph, from the sequence of static graphs generated by Algorithm~\ref{algo:robustGraph}, it is guaranteed to find a robust graph after at most $f+1$ reconfigurations.
\end{proof}

\subsection{Fanout vs Dependability Tradeoffs} 

In a perfect $m$-ary tree, we have  $
\approx \frac{N}{m}$ internal nodes. When using Algorithm~\ref{algo:robustGraph}, Optimal Conformity can only be achieved if $f<m-1$. For most systems, this is significantly smaller than the theoretical maximum number of faults of $\frac{N-1}{3}$. 
This results in an inherent tradeoff between load distribution and the maximum fault-tolerance that may be achieved.
%by Algorithm~\ref{algo:robustGraph}. 
Roughly, by setting a fanout of $m$, we reduce the load of the root by a factor of $m$, but we also limit the tolerated faults to $\approx \frac{m}{N}$.
%rwhen reducing the load of the root by a factor of $m$, we also limit the percentage of tolerated faulty nodes to $m\%$ of the system size. 
We argue that this tradeoff is reasonable in practice.  Consider for instance a system with 421 nodes. By selecting $m=20$, \thesystem can organize the nodes in a perfect tree where the leader has only $5\%$ of the work compared to a star topology, while still tolerating $20$ simultaneous process faults.

\subsection{Integration with HotStuff} 

The reconfiguration strategy described above can be applied to HotStuff as follows. In HotStuff, if no consensus is achieved after a timeout, each process compiles a \textit{new-view} message that includes the last successful quorum and sends it to the next leader. In turn, the leader candidate awaits for $2f+1$ \textit{new-view} messages and, depending on the collected information, either continues the work of the previous leader (if a block was previously locked) or proposes its own block (if no block had been locked yet). Similarly, in \thesystem, and upon timeout, each process invokes the \textsc{build} function of Algorithm~\ref{algo:robustGraph} to construct the next tree and sends the same \textit{new-view} message to the root of the new tree. The root then also awaits $2f+1$ \textit{new-view} messages before re-initiating the protocol.

\section{Pipelining and Piggybacking}
\label{sec:pipelining}

As described earlier, HotStuff requires four rounds of communication for each instance of consensus. If HotStuff waited for each instance to terminate before starting the next one, the system throughput would suffer significantly. Therefore, HotStuff relies on a pipelining optimization, where the $i+1$ instance of consensus is started optimistically, before instance $i$ is terminated.
As a result, at any given time, each process participates in multiple consensus instances. Furthermore, to reduce the number of messages, HotStuff piggybacks together messages from the multiple instances that run in parallel.

By following the same structure of HotStuff, \thesystem is amenable to the same optimization. However, because \thesystem uses a tree to disseminate and aggregate messages, the latency to terminate a given round, and hence a consensus instance, is substantially larger than in HotStuff. 
This opens the door for more aggressive pipelining, allowing \thesystem to execute more consensus instances in parallel. Next, we discuss how pipelining is applied by \thesystem, not only to mitigate the effects of the higher per round latency but also to improve the overall throughput by leveraging the parallelism enabled by the dissemination and aggregation tree.

\subsection{Pipelining in HotStuff}

We start by providing an overview of HotStuff's pipelining using as an example the same seven node system previously  introduced in Figure~\ref{fig:star}.
Figure~\ref{fig:hotstuff-pipelining} illustrates the execution of multiple rounds of consensus in HotStuff (each round is depicted in a different colour). 

Consider the first round (in yellow) that starts with the leader sending the block to all other processes. The overall time required qdepends on the size of the data being transmitted (for simplicity, let us assume it is dominated by the block size), the available bandwidth, and the total number of nodes. To conclude a round, the leader has to collect a quorum of signatures. These signatures start flowing towards the leader as soon as the first process receives the message from the leader and sends its reply back. Thus, in a given round, the dissemination and the aggregation procedures are partially executed in parallel.  Also, by the time the leader finishes transmitting the block to the last process, it might have already received $2f+1$ signatures (\textit{quorum received} arrow in the figure) and hence it is already processing them while the dissemination terminates.  Thus, as soon as the dissemination of the first round finishes the leader might already be able to start the second round of consensus. 

To implement pipelining, HotStuff optimistically starts a new instance of consensus by piggybacking the messages of the first round of the next consensus instance with the messages of the second round of the previous instance. Because HotStuff uses four rounds of communication, this process can be repeated multiple times, resulting in messages that carry information from up to four pipelined instances of consensus. In HotStuff the depth of the pipeline (i.e., the maximum number of consensus instances that run in parallel) is thus equal to the number of communication rounds. 

\begin{figure}
     \centering
	\includegraphics[width=0.8\columnwidth]{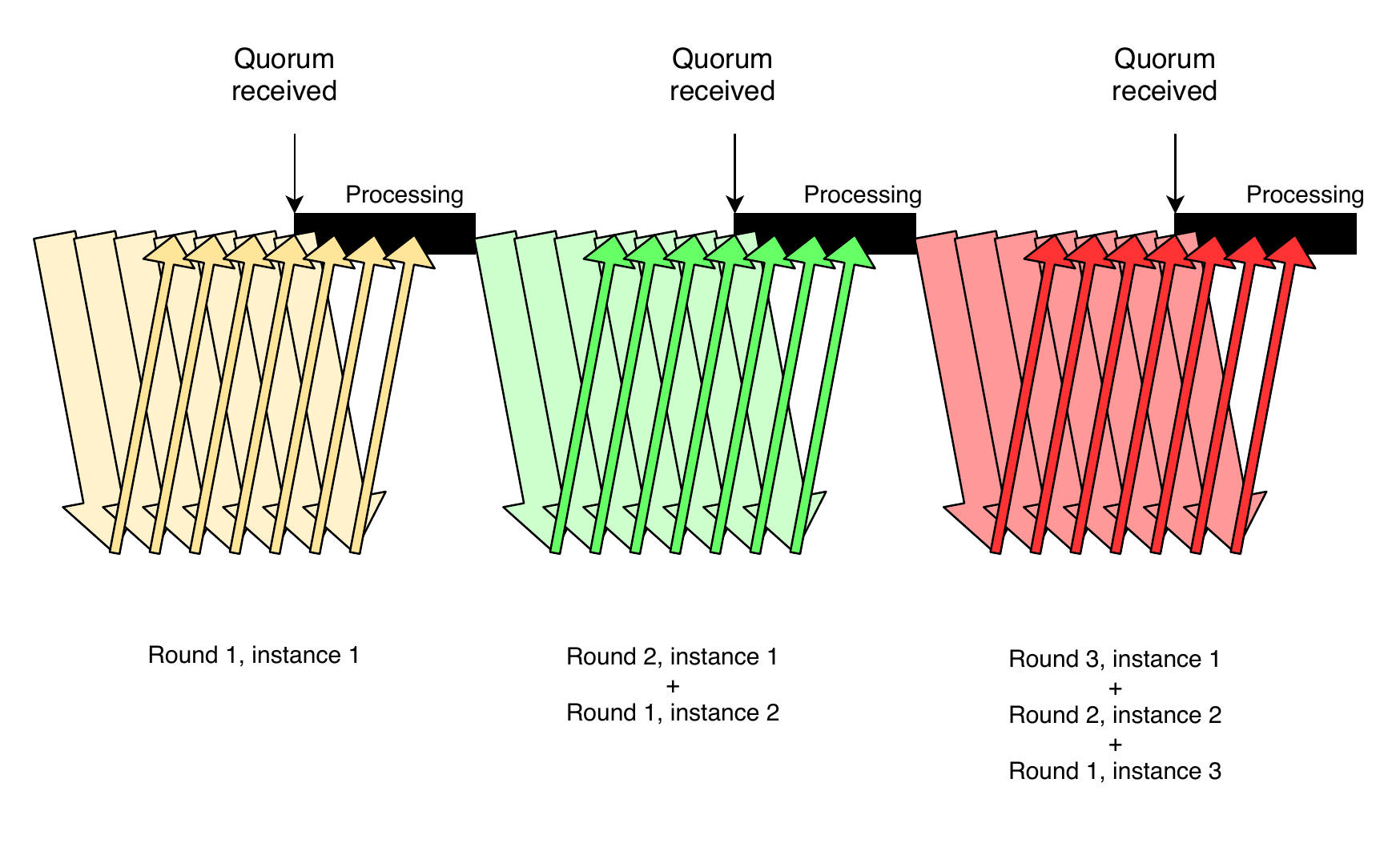}
	\caption{Pipelining in HotStuff}
  \label{fig:hotstuff-pipelining}
\end{figure}

\subsection{Pipelining in \thesystem}

In \thesystem, we enhanced the pipeling technique to fully leverage the load balancing properties of the tree. We illustrate this in Figure~\ref{fig:kauri-pipelining}. As with HotStuff, the aggregation phase of a given round starts before the dissemination phase is complete and, therefore, there is potential for parallelism between the two phases of the same round. However, in \thesystem, and in contrast to HotStuff, the root completes its dissemination long before it has collected a quorum of replies for that round. 

Like HotStuff, \thesystem implements pipelining by having the leader start a new instance of consensus shortly after the root finishes disseminating the messages of a previous instance.  The key difference is that \thesystem is able to start multiple consensus instances during the execution of a single round of consensus. In the example of Figure~\ref{fig:kauri-pipelining}, the leader is able to start 3 new instances during the execution of the first round of a given consensus instance.  Note that, in this example,  the messages from the second round of instance 1 are piggybacked with messages from the first round of instance 4, \emph{i.e.} a message carries information from consensus instances/rounds that are farther away in the pipeline. The increase in the pipeline depth allows for a higher degree of parallelism, and hence throughput, as we discuss in the following sections.

\begin{figure}
     \centering
	\includegraphics[width=0.6\columnwidth]{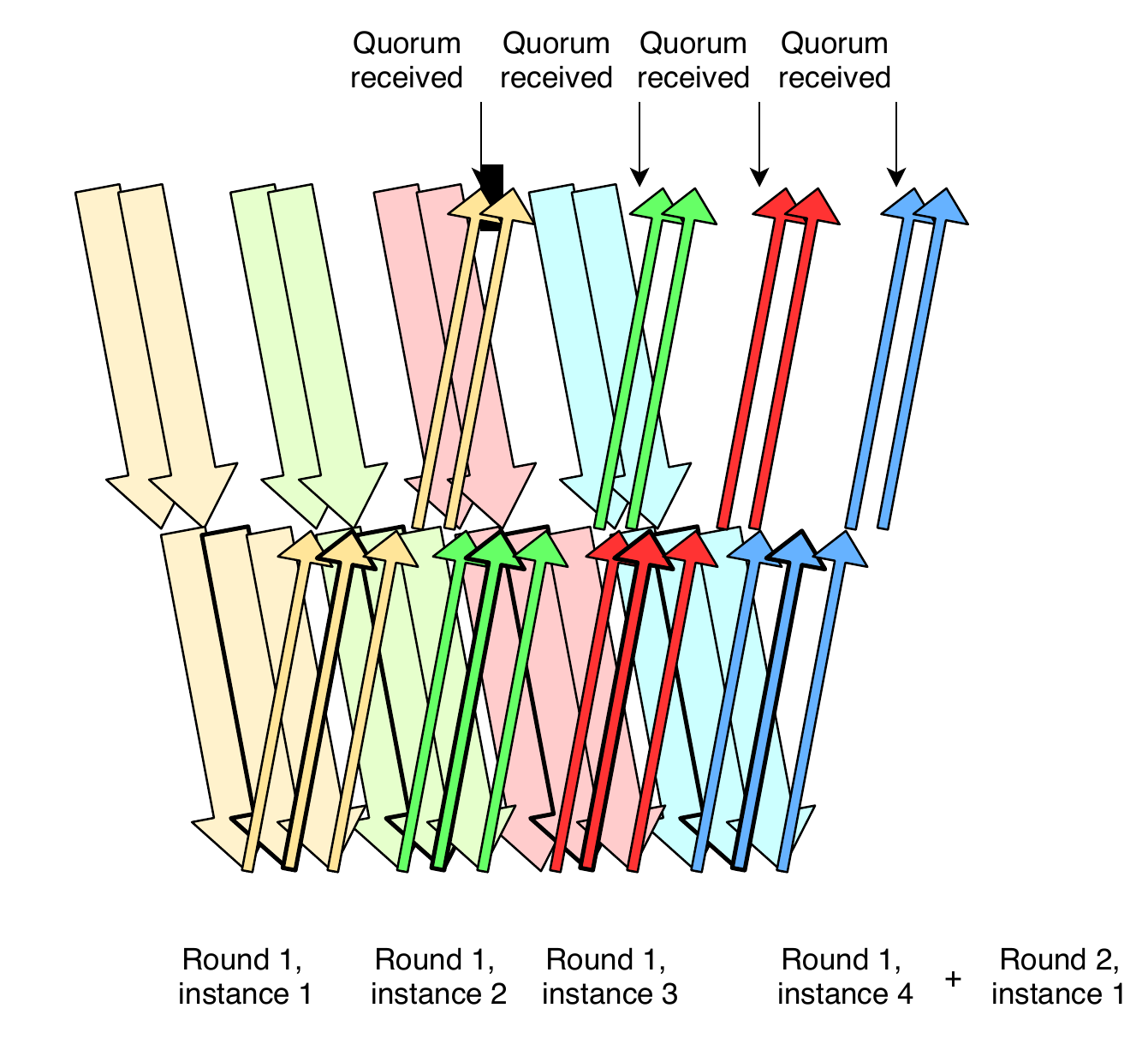}
	\caption{Pipelining in \thesystem}
  \label{fig:kauri-pipelining}
\end{figure}

\subsection{Pipelining Depth}
\label{sec:pipelining-throughput}

We now discuss the factors that affect  \thesystem's pipeline depth. Roughly speaking, the depth of the pipeline is given by the number of consensus instances the root node, or leader, is able to start during a round times the number of rounds of the protocol (four in our case, the same HotStuff).

To simplify the presentation, lets assume that the size of the messages sent by the root are dominated by the block size. 
Given the block size $B$ and the network bandwidth $b$ the time the root takes to send each message\footnote{For simplicity we assume the root is able to fully utilize the available bandwidth} is given $\frac{B}{b}$. 
With a fanout of $m$, the time the root requires to execute the dissemination phase is given by $\frac{mB}{b}$. After sending the last message, the root may still be required to process some of the replies, during some time $\phi$  before being ready to send a new block (illustrated by the black box in Figure~\ref{fig:hotstuff-pipelining}). We denote  the sum of the transmission plus processing time as the \textit{Busy Time (BT)}, which is given by  $\textit{BT} = \frac{mB}{b}+\phi$.

Let $T_{last}$ be the instant where the root sends the last message in a round. The root has to wait at least $h * \textit{RTT}$ to receive the last response, where \textit{RTT} is the average roundtrip time in the network and $h$ is the depth of the dissemination tree. We denote this as the \textit{Idle Time (IT)}, which is given by  $\textit{IT} = h \textit{RTT}$

Therefore, the number of parallel instances the root can start is bound by  $\frac{\textit{IT}}{\textit{BT}}$.  %Table~\ref{tab:examples} provides some examples of the impact of the height, fanout, block size, network bandwidth and RTT on the maximum number of parallel instances. 
To simplify this analysis, for now, we consider the processing time $\phi$ to be negligible but, in practice, and as we will discuss in the evaluation, the cost of cryptographic operations may also become a bottleneck in a real deployment. 
%The table first illustrates how additional pipelining does not yield any benefit on star topologies (first row) and also, how with increasing round trip time the idle time increases linearly and thus, the potential number of parallel instances also increases. It also shows that with a constant fanout, even with an increase of depth and scale, the pipelining may be used to compensate for the performance loss as well as that a decrease in fanout may allow to speedup the system significantly. 

%\begin{table}
%\centering
%{\scriptsize
%\caption{Maximum number of parallel consensus instances for different system configurations, assuming a  block size of 100K. The first row corresponds to a star and the remaining ones to trees.}
%\label{tab:examples}
%\begin{tabular}{|rrrrr|rrr|}
%\hline
%\textbf{N} & \textbf{h} & \textbf{m} &  \textbf{Band.} & %\textbf{RTT} & \textbf{Busy} & \textbf{Idle} & \textbf{Parallel} \\
%& &  &    (Mb/s)& (s) & \textbf{Time} & \textbf{Time} & \textbf{Instances} \\
%\hline

%100 & 2 & 99 & 25 & 0,20 & 0,396 & 0,2 & 0 \\ \hline
%100 & 3 & 10 & 25 & 0,20 & 0,04 & 0,4 & 10 \\
%100 & 3 & 20 & 25 & 0,20 & 0,08 & 0,4 & 5 \\
%421 & 3 & 20 & 100 & 0,20 & 0,02 & 0,4 & 20 \\
%421 & 3 & 20 & 100 & 0,10 & 0,02 & 0,2 & 10 \\
%1110 & 4 & 10 & 100 & 0,20 & 0,01 & 0,6 & 60 \\

%\hline
%\end{tabular}
%}
%\end{table}

%\subsection{Maximum Speedup}

In HotStuff, the leader needs to send $N-1$ messages before it can start a new instance. In \thesystem, the leader needs to send $m$ messages (where $m$ is the fanout of the tree) before it is ready to start a new instance. Therefore the theoretical speedup that can be achieved by \thesystem is bound by $\frac{N-1}{m}$. For instance, in a system of $400$ nodes, organized in a tree with fanout $20$, the speedup \thesystem can offer is $19.95$.

%Note that, the lower the fanout, the higher the speedup provided by \thesystem. However, as discussed in Section~\ref{sec:limitations}, increasing the number of internal nodes results in longer reconfiguration times.
%This exposes a tradeoff between speedup and reconfiguration time that should be taken into account when deploying \thesystem.
%Therefore, in the implementation we need to balance the speedup with the reconfiguration speed. 
%Recall that in deployments where \thesystem  must ensure Optimal Conformity, 
%In particular, when \thesystem is configured to be able to be reconfigured in optimal time, 
%the number of internal nodes is limited by $I =\frac{af+a}{f+1}$ (for a tree of depth 3, $m=I$). 
%Thus, in a scenario with  arbitrarily deep and perfectly balanced trees, $m^{h-1} = I$, and hence    the maximum speedup is bound by $\frac{N-1}{\root h-1 \of a}$.\mm{?}

\section{Implementation}
\label{sec:implementation}

The implementation of \thesystem is based on the HotStuff open source implementation available at \url{https://github.com/hot-stuff/libhotstuff}. We extended the existing implementation to support the
 \textsl{disseminate-to-followers} and \textsl{aggregate-at-leader} primitives as specified in Algorithm~\ref{algo:disseminate} and Algorithm~\ref{algo:aggregate}, respectively.
Moreover, we also added support for the \textit{bls} cryptographic scheme  by adapting the publicly available implementation used in the Chia Blockchain~\cite{bls,chia}. This allows internal nodes to aggregate and verify the signatures of their children, and thus balance the computational load.
Both the verification cost of the signature aggregates, and the size of aggregates is a small constant $\mathcal{O}(1)$ complexity, contributing to the overall efficiency of the implementation. Overall, these changes required the addition/adaptation of $\approx 1300$ loc.

%HotStuff implements pipelining by leveraging the unique blockchain properties. Thus, a new block is produced in each round of communication where the first quorum for block $i$ is implicitly the second quorum of block $i-1$ etc. Our pipelining extends this strategy, thus, depending on the number of pipelined blocks $p$ the quorum for block $i$ implicitly is also the quorum for block $i-1-p$.
 
To implement pipelining, we use an estimation of the parameters discussed above to compute the ideal time to start the dissemination phase for the next consensus instance. 
In the current implementation, we use a static pre-configured value but this could be automatically adapted at runtime, which we leave for future work.
%pAlthough the current implementation uses a static, pre-configured, pipelining depth, it is possible to extend \thesystem to auto-tune this parameter; this is left for future work.

%Following that, we've adjusted the block creation code of HotStuff to support pipelining by adjusting the block hierarchy to include the blocks that are currently being pipelined. Similarly, we've had to adjust the code that appends the block to the chain to prematurely append pipelined blocks after receiving a first quorum to be finalized at a later date. For this reason, we had to adjust the finalization code to trigger an additional pipelined block and finalize the block with the right offset (after sufficient quorums have passed).

\section{Evaluation}
\label{sec:evaluation}

In this section, we evaluate the performance of \thesystem in comparison to HotStuff. 
We assess the following: i) performance comparison of the different cryptographic schemes to establish a baseline; ii)  throughput comparison between \thesystem and HotStuff; iii) impact of the pipelining; and iv)  reconfiguration time. 

\subsection{Experimental Setup}

We use three system sizes with $N={100, 200, 400}$ processes and a fixed tree height $h=3$.
Rather than using some $N$ that yields perfect $m$-ary trees, which is unlikely in a real deployment, we instead deploy the above system sizes and distribute nodes among the internal nodes equally to approximate a balanced tree as follows: for $N=100$, $m=10$ for the root and $m=[8,9]$ for the other internal nodes; for $N=200$ the root's fanout is $m=15$, and for the other internal nodes $m=[12,13]$, and for $N=400$, $m=20$ for the root and $m=[18,19]$ for the remaining internal nodes.

All experiments were performed on the Grid'5000 testbed~\cite{grid5000}. Each machine comes with two Intel Xeon E5-2620 v4 8-core CPUs and 64 GB RAM. We used a total of 12 physical machines and evenly distribute the processes across the machines with the following restrictions: i) the leader always runs on a dedicated machine only co-located with its client process and ii) children and parent processes are always on different physical machines. Besides processing capacity, network capacity also plays a key role in the performance of the system.
To model the network characteristics, we use Linux's \textit{netem} to restrict bandwidth and latency among processes according to three different scenarios used in previous work and real uses-cases.

We consider three different deployment scenarios, namely the \textit{large-scale} deployment, the \textit{regional} deployment, and the  \textit{national} deployment.
The \textit{large-scale} deployment models a globally distributed blockchain as used in other works~\cite{motor, algorand, byzcoin} with $200ms$ roundtrip latency and $25Mb/s$ bandwidth.  The two other scenarios model reported industry use-cases~\cite{blockchain50}, in more limited geographical deployments for example for local supply-chain management, exchange of medical data, or digital currency (E-Euro, or E-Dollar).
The \textit{regional} deployment, captures a deployment in a large country or unions of countries, such as US or the EU, with $100ms$ roundtrip latency and $100Mb/s$ bandwidth.
Finally, the \textit{national} deployment models a setting where nodes are closer to each other, connected by a network  with $20ms$ roundtrip time and $100Mb/s$ bandwidth. 

We use a block size of $100Kb$, containing $400$ operations, as used for most scenarios in the original HotStuff's evaluation\cite{yin2019hotstuff}.
Note however that our work addresses geo-distributed deployments, which are significantly different than the local datacenter scenarios considered in~\cite{yin2019hotstuff}, that had a $1ms$ roundtrip latency and up to $8Gb/s$ of bandwidth usage.

\subsection{Effect of Cryptographic Operations in the Busy Time}

As discussed in Section~\ref{sec:pipelining}, the system throughput is limited by the Busy Time at the root, \emph{i.e.}, the time the root takes to send a block to its children, plus the additional time $\phi$ it requires to process the replies. Note that, in practice, communication and computation can partially overlap, and thus $\textit{BT} = \frac{mB}{b}+\phi$  provides just an approximation. In system with low bandwidth, \textit{BT} will be dominated by the communication, and in systems with low processing power it may be dominated by $\phi$. Next, we evaluate the values of $\phi$ for different  cryptographic schemes: i) \textit{secp256k1} used by HotStuff, ii) \textit{bls}~\cite{bls,chia} used by \thesystem
 and iii) \textit{gpg}~\cite{gnupg}, a classic cryptographic scheme.

Figure~\ref{fig:crypto} shows the cost of a single sign and verify operation averaged over 10 runs for each scheme.
As one can see, \textit{gpg} is substantially slower than \textit{bls} and \textit{secp256k1}, with the latter clearly outperforming \textit{bls}.
However, while the cost of \textit{gpg} and \textit{secp256k1} grows linearly with the number of signatures, \textit{bls} amortizes this cost through aggregation. 
The results for \textit{bls} aggregation are show in Figure~\ref{fig:blsagg} and can be decomposed into aggregation of multiple signatures to create one single aggregated signature and the aggregation of public keys to create an aggregated public key that is required to verify aggregated signatures. Interestingly the aggregation cost decreases quickly with the number of elements (number of signatures, public keys) while the verification cost remains constant (Figure~\ref{fig:crypto}). In contrast, the verification cost of the other schemes grows linear with the number of signatures.
This means that after $\approx 55$ signatures (\emph{i.e.} inputs of processes), the overall cost of \textit{secp256k1} surpasses that of \textit{bls}.

\begin{figure}
	\centering
	\subfloat[Single operation\label{fig:crypto}]{\includegraphics[width=0.49\linewidth]{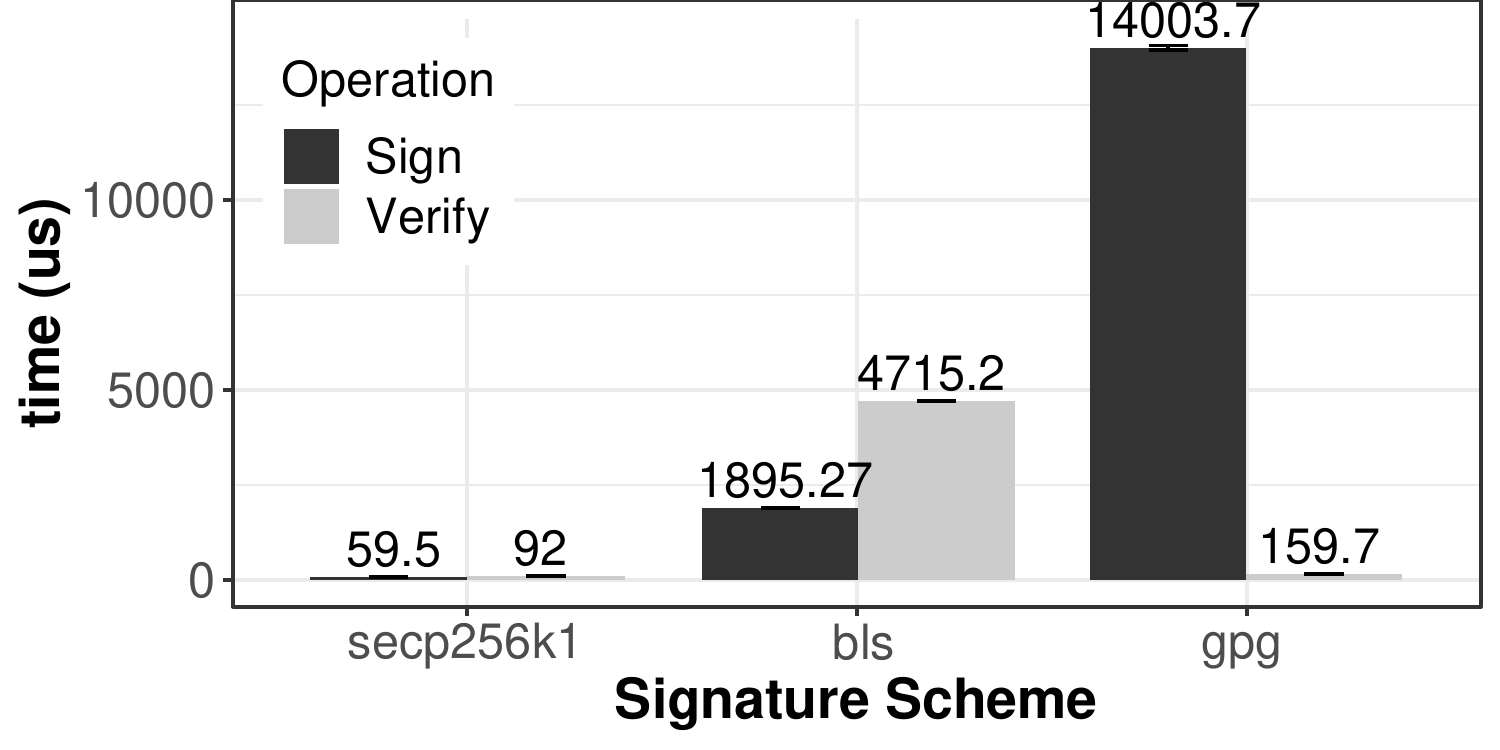}}\,
	\subfloat[\textit{bls} Aggregation\label{fig:blsagg}]{\includegraphics[width=0.49\linewidth]{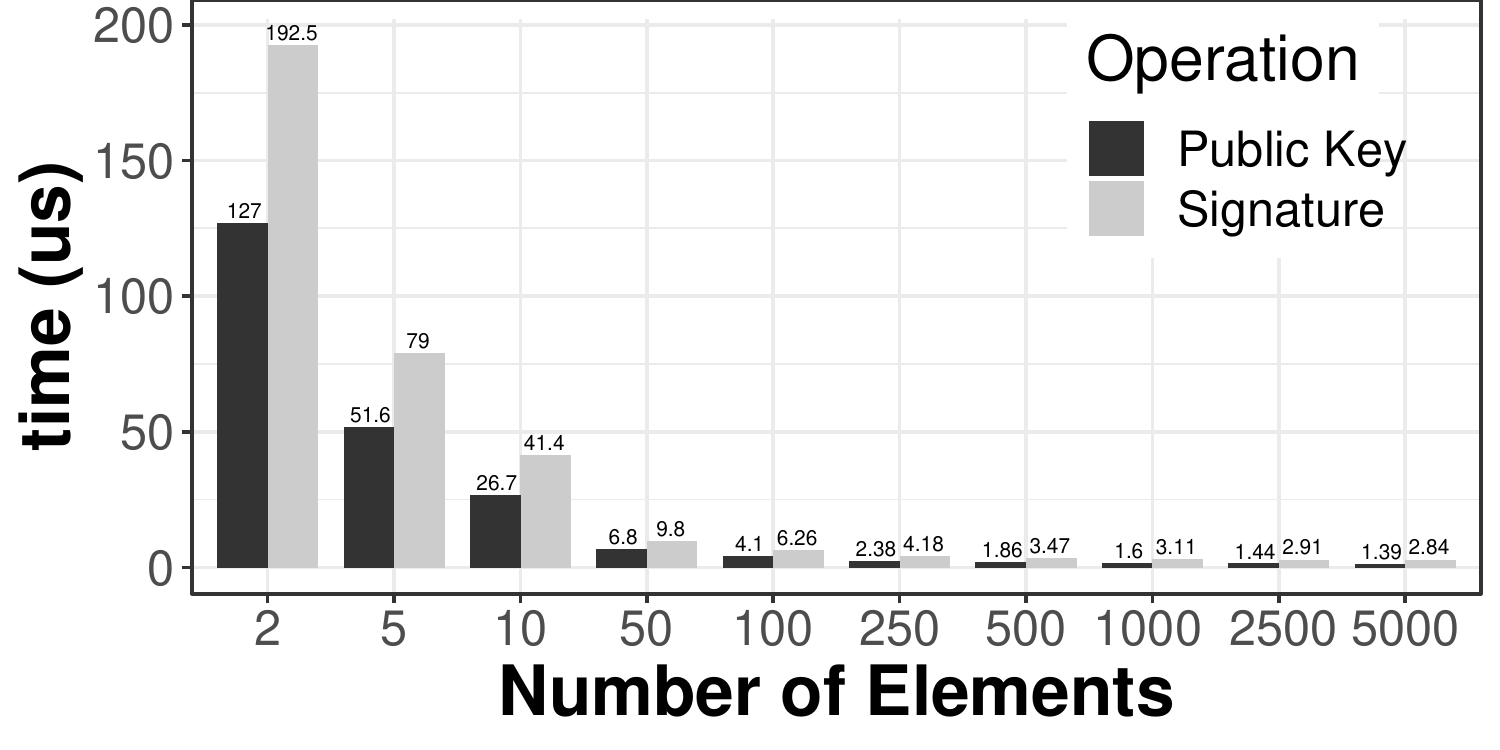}}
    \caption{Cost of cryptographic operations (Sign and Verify)}
    \label{fig:cryptocost}
\end{figure}

\subsection{Pipelining}

We now assess the effects of pipelining on speedup.
Table~\ref{table:phi} shows the results for different system sizes in the \textit{large-scale} scenario.
The Idle Time (IT) is fixed for each system (Section~\ref{sec:pipelining-throughput}), and we measured the two components of \textit{BT} to estimate the pipelining stretch $(\frac{IT}{BT})$ and maximum speedup.
Note, that the reported pipelining stretch, is a multiplicative factor that is applied to the base pipelining of HotStuff (4 consensus instances). 

As it is possible to observe, even though the base $\phi$ of \thesystem is higher than HotStuff, it grows much slower with the number of processes and, thus, provides better opportunities to increase the pipelining depth both in terms of its computational cost as well as in terms of bandwidth consumption.
In fact, the opportunities for extra pipelining in HotStuff are negligible, showing that further increasing it does not result in additional performance benefits.
In contrast, in \thesystem, as we increase the number of processes, pipelining can be used to create a significant speedup compared to HotStuff.
Note that the estimated speedup of \thesystem differs from the theoretical results of Section~\ref{sec:pipelining-throughput}. For instance, for $N=400$ and $m=20$, the theoretical speedup is $19.95$ while the estimated speedup in Table~\ref{table:phi} is $32.0$.
This can be explained as follows.
The theoretical analysis considers that the dominating factor in the message size is the block size. 
However, the use of \textit{secp256k1} signatures by HotStuff results in a growing overhead as the system size increases ($2f+1$ signatures each of size $64$ bytes) when compared to \textit{bls} which remains constant ($m$ signatures of $96$ bytes each).
Note that this is an inherent limitation of  \textit{secp256k1} cryptographic scheme that only gets more severe as the system size grows (assuming a constant block size). We further study the impact of this in the next section.

\begin{table}
\begin{center}
{\scriptsize
\begin{tabular}{l|r|r|r|r|r|r|r|r}
\hline
\# & root's & Idle & $\phi$ & & Busy & \multicolumn{2}{c|}{Pipelining} & Esti. \\\cline{7-8}
nodes & fanout  & Time &  (ms) & $\frac{mB}{b}$ & Time & base & stretch & speedup \\
 & $m$   &  &   & &  & &  &  \\\hline
\multicolumn{8}{c}{HotStuff }  \\ \hline
100 & 99 & 200 & 10,58 & 550 & 561 & 4 & 1 & - \\
200 & 199 & 200 & 19,66 & 1375 & 1395 & 4 & 1  & -  \\
400 & 399 & 200 & 36,91 & 3843 & 3880 & 4 & 1  & - \\\hline

\multicolumn{8}{c}{\thesystem} \\ \hline
100 & 10 & 400 & 31,98 & 42 & 74 & 4 & 6 & 7.5 \\
100 & 20 & 400 & 32,83 & 83 & 116 & 4 & 4 & 4.8 \\
200 & 20 & 400 & 37,33 & 83 & 120 & 4 & 4 & 11.6 \\
400 & 20 & 400 & 37,90 & 83 & 121 & 4 & 4 & 32.0 \\\hline

\end{tabular}
}
\end{center}
\caption{Pipelining depth for the \textit{large-scale} deployment with different system sizes and estimated speedup vs HotStuff.}
\label{table:phi}
\end{table}

%We now show the effect of pipelining and compare the observed behaviour with the behaviour predicted by our simplified performance model on the \textit{large-scale} deployment. Table~\ref{table:phi} shows depicts the predicted maximum pipeling depth for different fanouts of the tree, both when considering only bandwidth and when considering only cpu consumption. For instance, for fanout 10, when considering the bandwidth limitation alone, the maximum speedup over HotStuff is 10x. However, this configuration is cpu-bounded and, given that $\phi \approx 190$, the  expected maximum pipeline depth is close to 5. When using a fanout of 20, the predicted maximum depth is approximately 5, considering either bandwidth or cpu.

We now analyze the effect of pipelining on throughput and compare it with the estimated speedup values of Table~\ref{table:phi}. This experiment highlights that the pipelining stretch supported by \thesystem can bring substantial performance benefits. 
Results are show in Figure~\ref{fig:figpiperes} for $N = 100$ and  $m = 10,20$ in the \textit{large-scale} scenario.
For $m=10$ we can see a linear increase in throughput up to $7x$ after which performance starts to degrade (the stretch in Table~\ref{table:phi} is $6x$ for this configuration).
Note, that the measured speedup is slightly higher than the estimated one due to some overlap in communication ($\frac{mB}{b}$) and computation ($\phi$) not captured in our simplified model.
For $m=20$, we observe a linear growth until $5x$ after which throughput also degrades albeit slower than for $m=10$, until a substantial drop at $8x$ in both configurations due to saturation.
%As is is possible to observe,  by increasing the pipelining stretch (from the base stretch of 1, that uses the same pipeling depth as HotStuff) we observe an almost linear increase in throughput up to 
%$7x$ where the system becomes saturated and performance decreases.
%Finally, while the second curve (fanout 20) increases until 7x, it saturates already at 5x which can be observed from the sublinearity after 5x. 

%The same does not happen with fanout 10 after 8x, since the network library used in HotStuff has difficulties coping with the high pipelining/latency scenario.

\begin{figure}
\begin{center}
	\includegraphics[width=0.6\columnwidth]{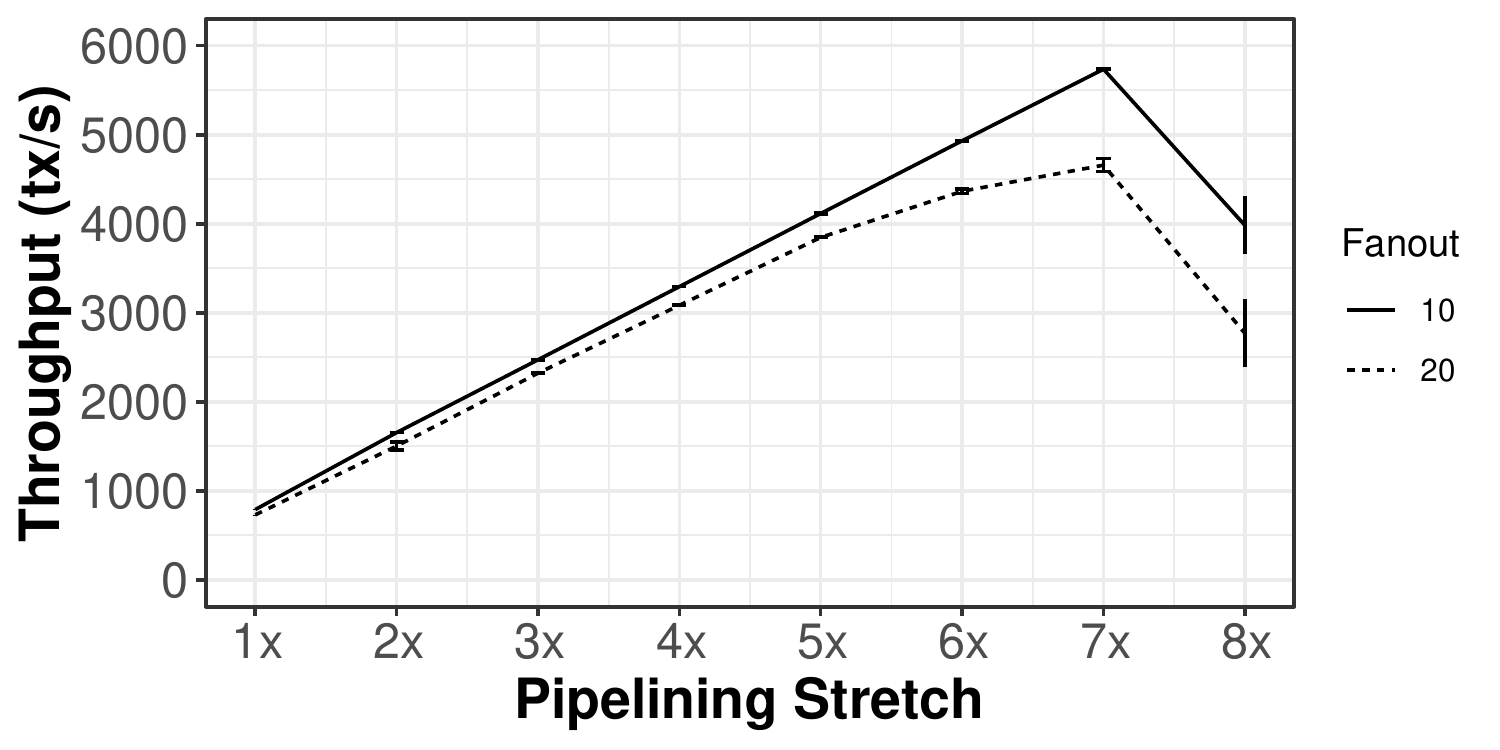}
\end{center}
\caption{Effect of pipelining stretch on \thesystem's throughput for $N=100$ and $m=10,20$.}
\label{fig:figpiperes}
\end{figure}

%We now show the effect of pipelining and compare the observed behaviour with the behaviour predicted by our simplified performance model on the \textit{large-scale} deployment. Table~\ref{table:phi} depicts the predicted maximum pipeling depth for different fanouts of the tree. For instance, for fanout 10, the maximum pipelinig depth is 5,4x. When using a fanout of 20, the predicted maximum depth is approximately 3,4x.

%Figure~\ref{fig:figpiperes} shows the experimental results obtained when using two different fanouts. As may be observed, we were able to improve the system throughput, almost linearly, by increasing the pipelining depth up to the predicted value. In fact, the real system scales slightly above the predicted limit due to some overlap between the communication and computation cost ($\frac{mB}{b}$ and $\phi$). However, when the maximum pipelining depth is exceeded, the system saturates and the performance falls abruptly since the load exceeds the system capacity. The resulting instability of the system when operating well above its limits is also visible in the throughput variance.

\subsection{Throughput}

Next, we assess the performance of \thesystem both with full pipelining which we call \thesystemfull and with the base pipelining which we call \thesystemcrippled. We compare it with the regular HotStuff which uses the \textit{secp256k1} cryptographic scheme (HotStuff-secp) and also with an adapted version that uses \textit{bls} signatures (HotSuff-bls).
Figure~\ref{fig:throughput} depicts the results of this experiment for the different network scenarios. 
While the variants of HotStuff offer similar or slightly better performance with smaller system sizes when compared to \thesystemcrippled, this rapidly changes as the number of processes grows and eventually \thesystemcrippled takes over. This is explained by the fact that with smaller system sizes the additional latency of the tree is higher than the bandwidth and computing bottleneck of the star. However, that shifts as the number of processes increases.

\begin{figure*}[t]
	\centering
	\subfloat[large-scale ($200ms-25Mb/s$)\label{fig:scenB}]{\includegraphics[width=0.27\linewidth]{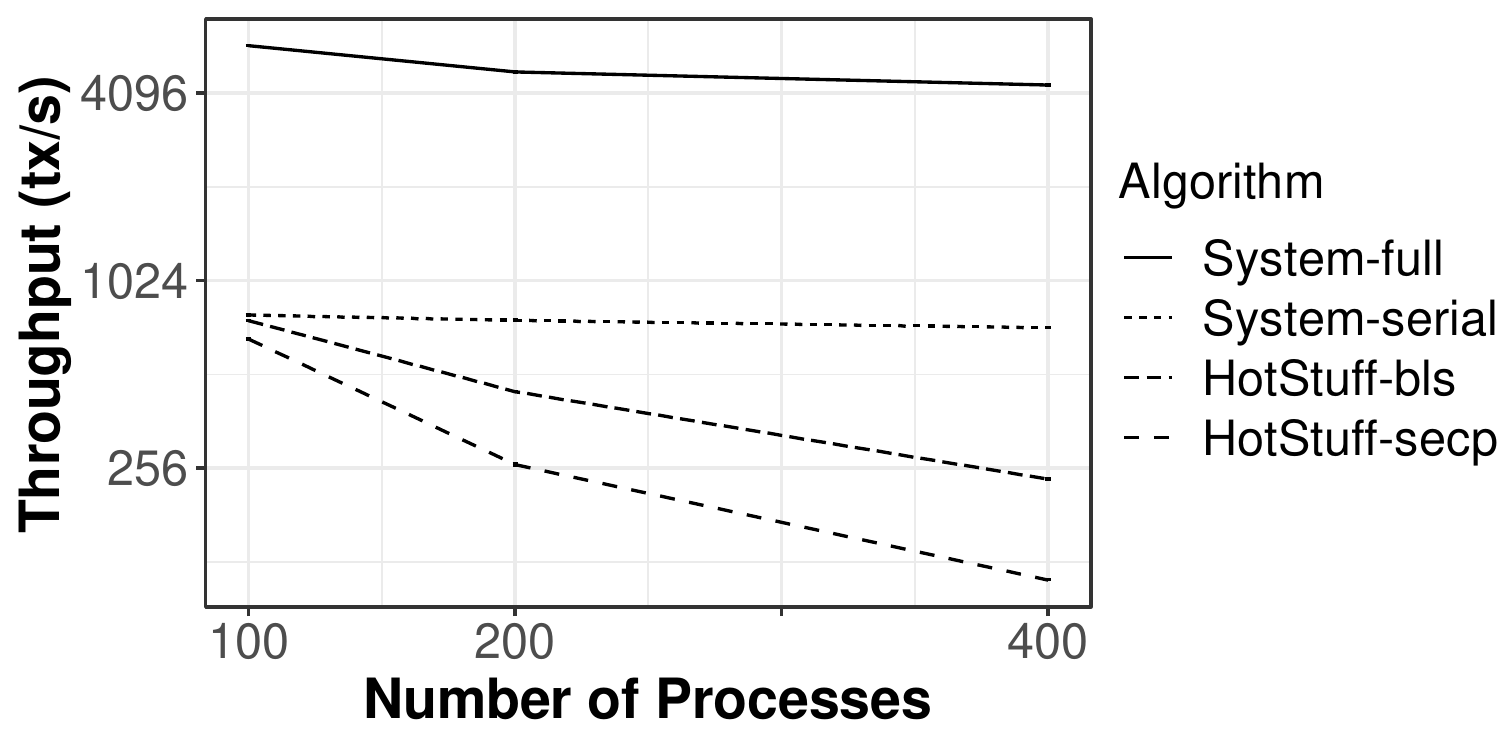}}\,\,\,
	\subfloat[regional ($100ms-100Mb/s$)\label{fig:scenA}]{\includegraphics[width=0.27\linewidth]{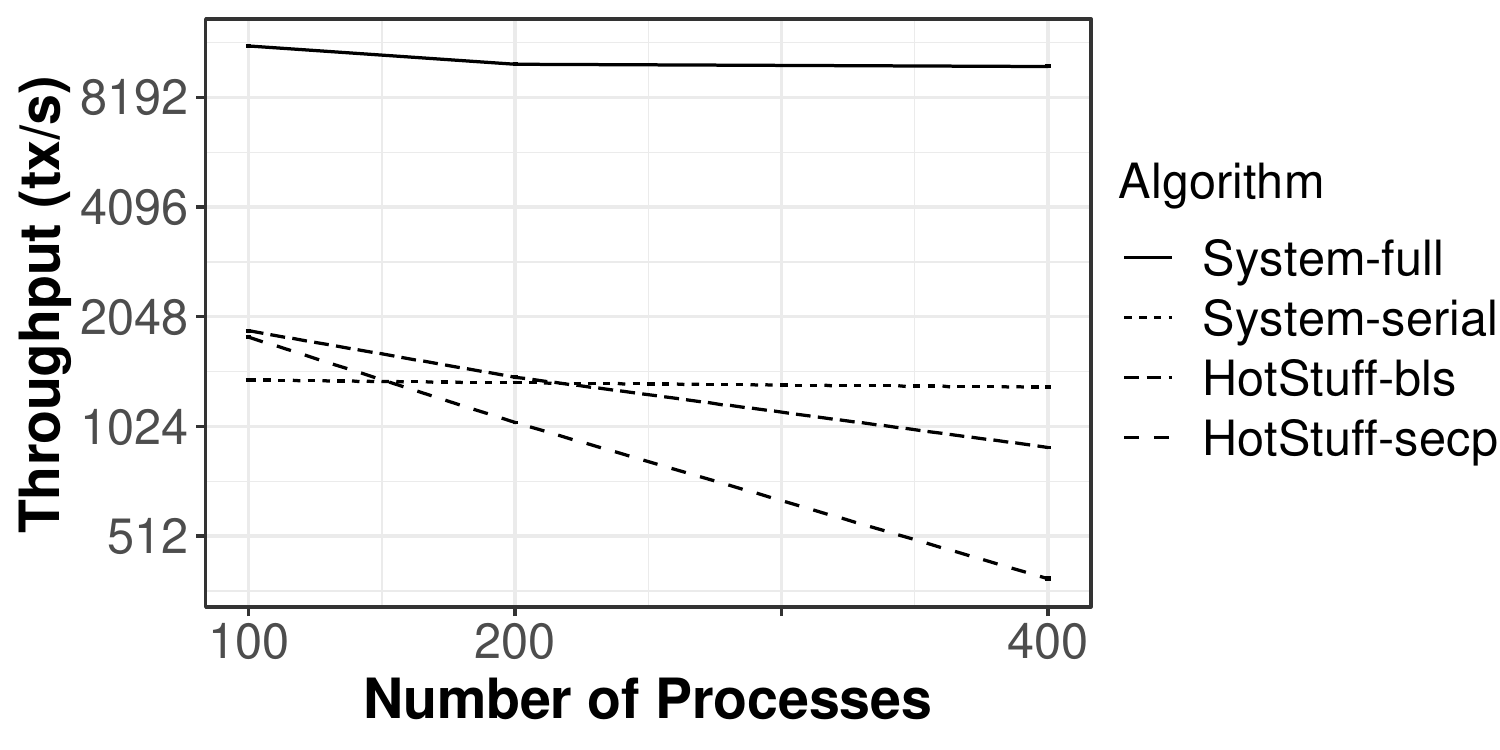}}\,\,\,
	\subfloat[national ($20ms-100Mb/s$)\label{fig:scenC}]{\includegraphics[width=0.27\linewidth]{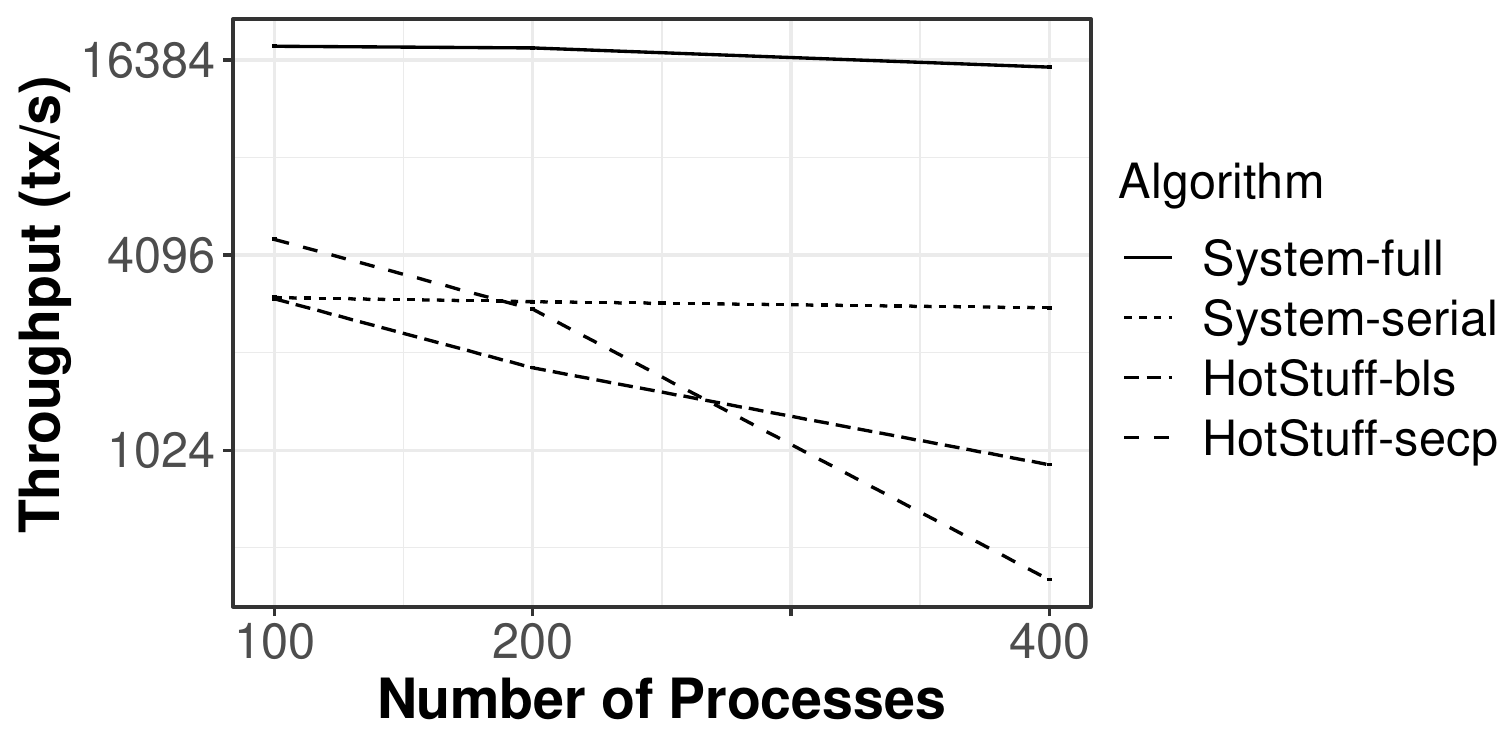}}
	\caption{Throughput for different latency/bandwidth configurations.}
	\label{fig:throughput}
\end{figure*}

Finally, the throughput of \thesystemfull is significantly higher than both variants of HotStuff.
As argued throughout the paper, this stems from the fact that the leader of the star becomes a communication and/or computation bottleneck depending on the scenario. This does not happen in \thesystemfull, which is able to fully exploit the load balancing capacity of all the internal nodes.
In fact, by leveraging pipelining, \thesystemfull is able to overcome the additional latency stemming from the use of a tree and %Moreover, \thesystemfull leverages pipelining to compensate for the higher latency of the tree allowing it to 
hence process a large number of operations per second where both the computational as well as bandwidth consumption are equally distributed over the tree.
This results in \thesystemfull offering up to 17x the throughput of HotStuff-bls, our adaptation, and up to 38x the throughput of vanilla HotStuff-secp.

%outlined before, in all scenarios, due to the large number of processes, the star bottlenecks at the bandwidth, significantly limiting its scalability.
%We may also observe this from the fact that the implementation of the star that uses \textit{bls-signatures} outperforms the one using \textit{secp} even though the computational cost is significantly higher. 

\subsection{Leader Replacement}

Finally, we evaluate the reconfiguration time in two faulty scenarios: i) one fault at the leader, and ii) three consecutive leader faults.
To calibrate the fault detection timeout, we started with a large value and experimentally decreased it until we found that further decreases would lead to spurious reconfigurations in a stable system.
This resulted in a timeout of $0.3ms$ for \thesystem and $0.75ms$ for HotStuff-bls. 
Figure~\ref{fig:reconfig} presents the results for a system with $N=100$ in the \textit{large-scale} scenario for \thesystemfull and HotStuff-bls.
We run each system for $100$ seconds, inject the fault at $60$ seconds and measure the impact on the throughput.
As we can observe, both HotStuff-bls and \thesystemfull can  recover, in both faulty scenarios, to the throughput before the fault in $5$ seconds for one fault and $8$ seconds for three consecutive faults.
The recovery time of \thesystemfull is thus comparable to HotStuff-bls despite the use of a more complex topology, hence stressing the effectiveness of our reconfiguration strategy.

%how \thesystem compares in terms of reconfiguration performance. For this purposed we have injected faults in the system to force one or more reconfigurations. In the both cases the current leader fails, but in one case, the next leader is correct and, in the other case, the next two  leaders are also faulty. For both Kauri and HotSuff we have experimentally computed that lowest timeout possible that would not lead to spurious reconfigurations in steady-state. This resulted in a timeout value of $0.3ms$ for \thesystem and $0.75ms$ for HotStuff. The results of those experiments are depicted in Figure~\ref{fig:reconfig}. The results show that \thesystem recovers quickly from the injected fault.  Also, shortly after recovery,  the pipelining is filled again, bringing the system back to the original throughput.

\begin{figure}[t]
	\centering
	\subfloat[1 fault\label{fig:reconfig1}]{\includegraphics[width=0.48\columnwidth]{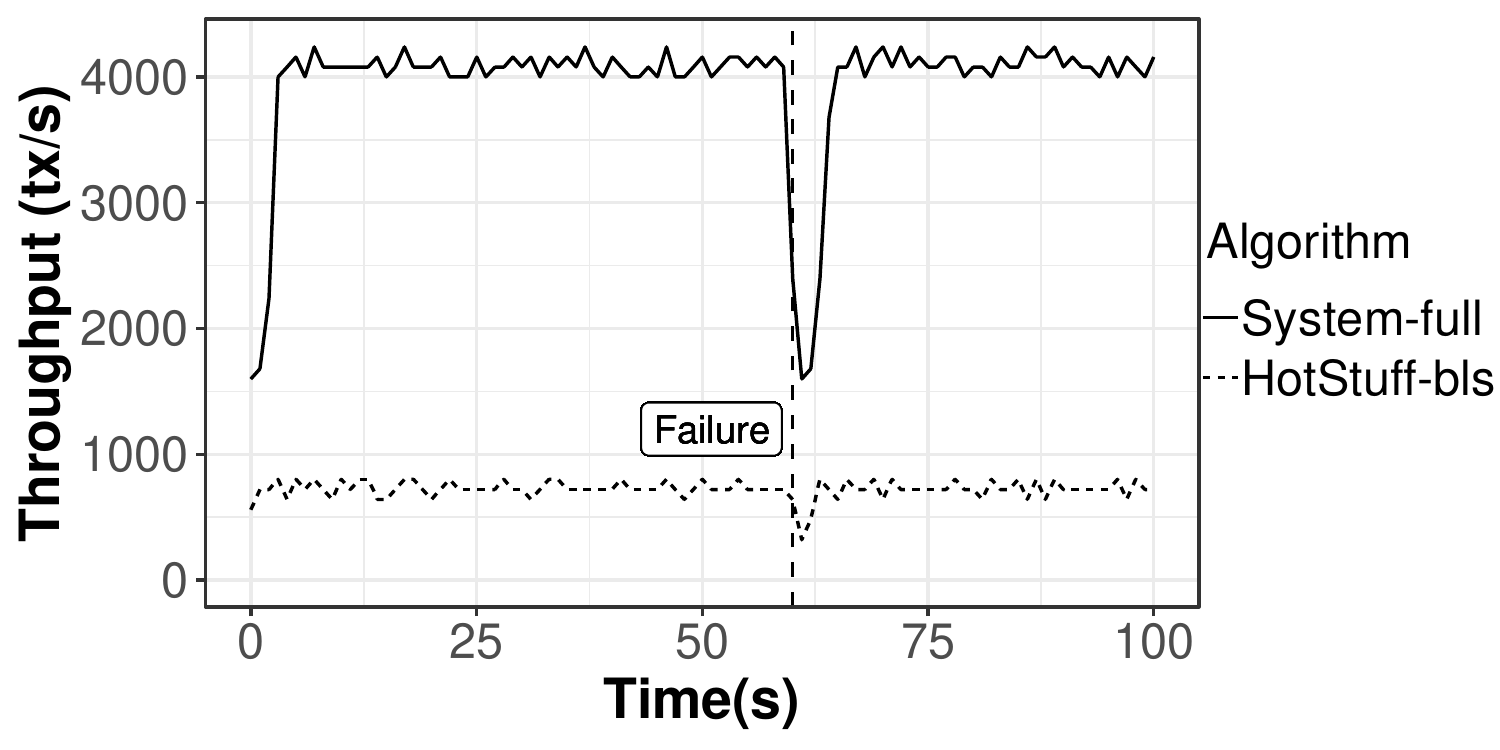}}\,
	\subfloat[3 simultaneous faults\label{fig:reconfig2}]{\includegraphics[width=0.48\columnwidth]{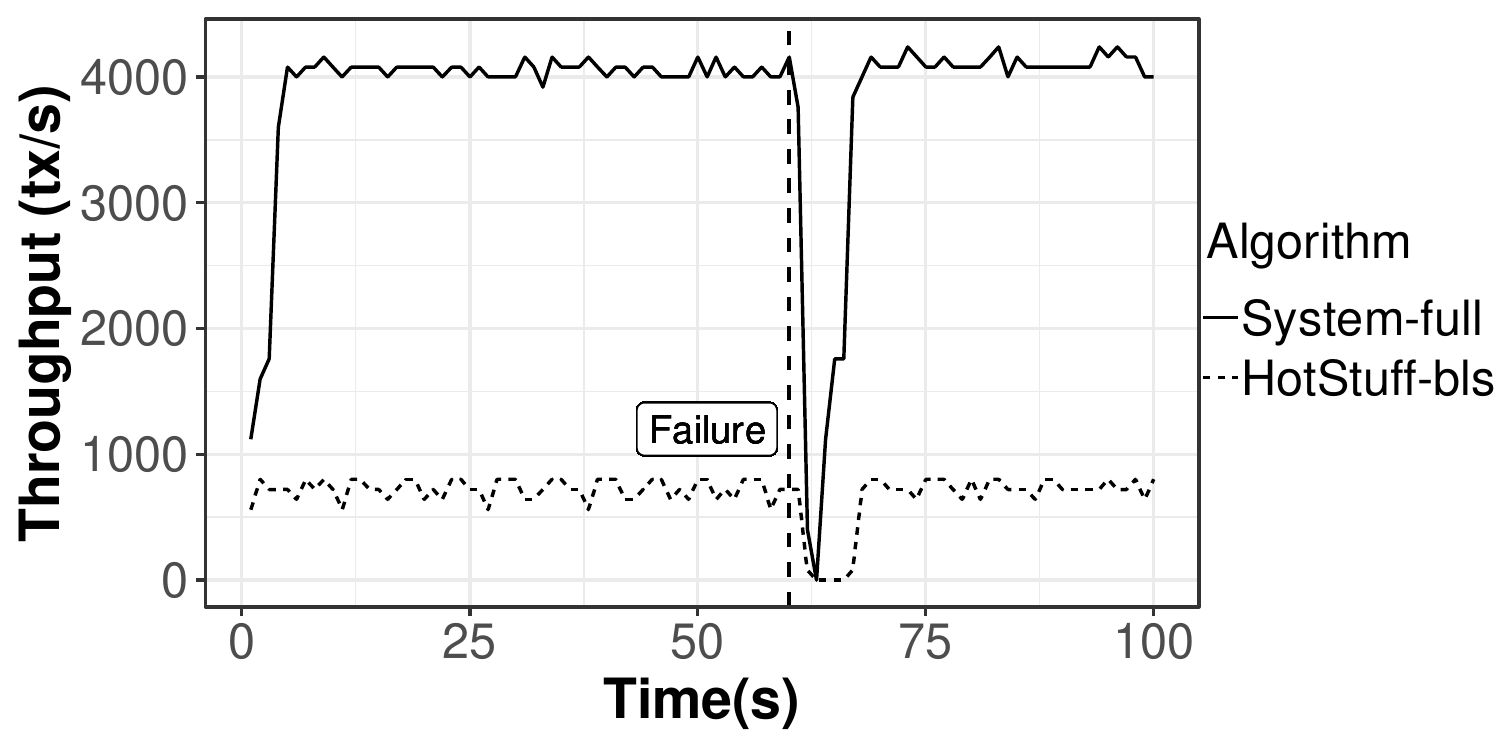}}
	\caption{Reconfiguration after failure}
\label{fig:reconfig}
\end{figure}

\section{Conclusion}
\label{sec:conclusion}

The growing interest in permissioned blockchains led to the appearance of novel consensus algorithms that aim at improving classic PBFT-based algorithms. However, state-of-the-art approaches still suffer from important limitations: some concentrate work on the leader, that becomes a bottleneck, others distribute the load but cannot recover from failures quickly. \thesystem overcomes these limitations by using a tree to distribute the communication and processing load, while ensuring recovery from  faults within an optimal number of steps.  Our distribution of work also allows to increase the number of instances that can run in parallel without overloading the system.  Our implementation and evaluation in realistic scenarios with up to 400 processes, shows that \thesystem can not only achieve up to 38x the throughput of HotStuff but also offer better scalability characteristics when increasing the number of processes. This shows that the mechanisms proposed here are a powerful tool to address the scalability limitations of permissioned blockchains.

%This was achieved mainly by departing from the all-to-all clique topology of PBFT and leveraging instead star or tree topologies. However, such state-of-the-art approaches still suffer from important limitations. On the one hand, while the use of a star increases efficiency by requiring fewer messages, it places a processing and communication bottleneck in the leader which limits scalability. On the other hand, trees are effective in balancing the load but degenerate to a star or clique upon failures.

%\thesystem leverages parallelism to significantly increase the throughout of consensus with large number of participants. In strike contrast with previous attempts of using trees to support the execution of consensus, \thesystem is able to reconfigure in an optimal number of steps. 

\bibliographystyle{IEEEtran}
\bibliography{bib-short}

\end{document}